\documentclass{LMCS}

\def\dOi{10(1:19)2014}
\lmcsheading%
{\dOi}
{1--33}
{}
{}
{Oct.~29, 2010}
{Mar.~31, 2014}
{}

\subjclass{F.2 Analysis of Algorithms and Problem Complexity,
F.4 Mathematical Logic and Formal Languages}

\ACMCCS{[{\bf Theory of computation}]: Design and analysis of
  algorithms; Computational complexity and cryptography---Complexity
  theory and logic; Logic; Formal languages and automata theory}

\usepackage[english]{babel}
\usepackage{amssymb}
\usepackage{amsmath}
\usepackage{color}
\usepackage{cite}
\usepackage{graphicx}
\usepackage{float}
\usepackage{subfig}
\usepackage{tikz}
\usepackage{xspace}
\usepackage{ulem}
\usepackage{hyperref}

\usetikzlibrary{positioning,fit}

\theoremstyle{plain}

\DeclareSymbolFont{letters}{OML}{cmbboard}{m}{it} 

\newcommand{\ncp}{\textnormal{\texttt{ncp}}}

\newcommand{\NE}{\hspace{-0.4em}&\hspace{-0.4em}}
\newcommand{\NR}{\\}
\newcommand{\minfty}{-\infty}

\newcommand\zero{\mathop{\mathsf{zero}}}
\newcommand\BITS{\text{$\TRS{R}_{\m{bits}}$}\xspace}

\newcommand\Nat{\ensuremath{\mathbb{N}}}
\newcommand\NAT{\Nat}
\newcommand\Arctic{\ensuremath{\mathbb{A}}}
\newcommand\ALG[1]{\ensuremath{\mathcal{#1}}}
\newcommand\SIG[1]{\ensuremath{\mathcal{#1}}}
\newcommand\VAR[1]{\ensuremath{\mathcal{#1}}}
\newcommand\TERMS[2]{\mathcal{T}(#1,#2)}
\newcommand\GTERMS[1]{\mathcal{T}(#1)}
\newcommand\CGTERMS[1]{\mathcal{T}_\mathsf{Con}(#1)}
\newcommand\FVTERMS{\TERMS{\SIG{F}}{\VAR{V}}}
\newcommand\FTERMS{\GTERMS{\SIG{F}}}
\newcommand\TRS[1]{\ensuremath{\mathcal{#1}}}
\newcommand\TA[1]{\ensuremath{\mathcal{#1}}}

\newcommand\REL[2]{{\text{\ensuremath{#1 \kern0em/\kern0em #2}}}}

\newcommand\WDP{\mathsf{WDP}}
\newcommand\UR{\mathsf{UR}}

\newcommand\linear{\ensuremath{\OO(n)}\xspace}
\newcommand\quadratic{\ensuremath{\OO(n^2)}\xspace}
\newcommand\cubic{\ensuremath{\OO(n^3)}\xspace}
\newcommand\poly{\ensuremath{\OO(n^k)}\xspace}
\newcommand\avg{time\xspace}

\newcommand\direct{\ensuremath{\mathsf{direct}}\xspace}
\newcommand\modular{\ensuremath{\mathsf{modular}}\xspace}

\newcommand\MATCH{\mathsf{match}}
\newcommand\MATCHRT{\text{$\MATCH$-$\mathsf{RT}$}}
\newcommand\RAISE{\mathsf{raise}}
\newcommand\BASE{\mathsf{base}}
\newcommand\LIFT{\mathsf{lift}}
\newcommand\HEIGHT{\mathsf{height}}
\newcommand\MULEXT{\mathsf{mul}}
\newcommand\MUL{\mathcal{M}\mathsf{ul}}
\newcommand\DROP{\mathsf{drop}}
\newcommand\MFUN{{\mathcal{H}}}
\newcommand\Lg[1]{\mathcal{L}(#1)}
\newcommand\Succ[3][\to]{{#1_{#2}^*}(#3)}
\newcommand\Mul[1]{\MUL(#1)}
\newcommand\MFun[1]{\MFUN(#1)}
\newcommand\Raise[1]{\RAISE(#1)}
\newcommand\match[1]{\MATCH(#1)}
\newcommand\matchRT[3]{\MATCHRT^{#3}(\REL{#1}{#2})}
\newcommand\base[1]{\BASE(#1)}
\newcommand\lift[2]{\LIFT_{#2}(#1)}
\newcommand\height[1]{\HEIGHT(#1)}
\newcommand\lhs[1]{\mathsf{lhs}(#1)}
\newcommand\mge{\succeq_\MULEXT}
\newcommand\mgt{\succ_\MULEXT}
\newcommand\relgt{\succ}
\newcommand\relge{\succeq}
\newcommand\drop[2]{\DROP_{#2}(#1)}

\newcommand\TRSALLNUM{2132\xspace}
\newcommand\TRSDCNUM{1172\xspace}
\newcommand\TRSRCNUM{1339\xspace}
\newcommand\TRSRCNDNUM{910\xspace}
\newcommand\Var{\mathsf{Var}}
\newcommand\Pos{\mathsf{Pos}}
\newcommand\FPos{{\Pos_\mathcal{F}}}
\newcommand\Fun{\mathsf{Fun}}
\newcommand\DFun{\mathsf{Def}}
\newcommand\CFun{\mathsf{Con}}

\renewcommand\max{\mathsf{max}}
\renewcommand\min{\mathsf{min}}
\renewcommand\sup{\mathsf{sup}}
\newcommand\BM{\begin{pmatrix}}
\newcommand\EM{\end{pmatrix}}
\newcommand\OO{\mathcal{O}}
\newcommand\m[1]{\mathsf{#1}}
\newcommand\app{\mathbin{\circ}}
\newcommand\dl[2]{{\mathsf{dh}(#1,#2)}}
\renewcommand\dh[2]{\dl{#1}{#2}}
\newcommand\cp[3]{\mathsf{cp}_{#3}(#1,{#2})}
\newcommand\dc[2]{{\mathsf{dc}(#1,#2)}}
\newcommand\rc[2]{\mathsf{rc}(#1,#2)}

\newcommand\MINISMT{\ensuremath{\mathsf{\mbox{$\mathsf{Mini}$}Smt}}\xspace}

\newcommand\tpdb[1]{\ensuremath{\mathsf{#1}}\xspace}

\newcommand\rto{\xrightarrow{\smash{
 \raisebox{-0.15em}{$\scriptstyle\mathsf{r}$}}}}

\newcommand\CETA{\textsf{CeTA}}
\newcommand\TTT{%
 \ensuremath{\mathsf{T\kern-0.2em\raisebox{-0.3em}%
 {$\mathsf{T}$}\kern-0.2emT}}\xspace%
}
\newcommand\TTTT{%
 \ensuremath{\mathsf{T\kern-0.2em\raisebox{-0.3em}%
 {$\mathsf{T}$}\kern-0.2emT\kern-0.2em%
 \raisebox{-0.3em}2}}\xspace%
}
\newcommand\CAT{\ensuremath{\mathsf%
 {C\kern-0.25em\raisebox{0.14em}{$\mathsf{a}$}\kern-0.2emT}}\xspace%
}
\newcommand\TCT{\ensuremath{\mathsf%
 {T\kern-0.2em\raisebox{-0.3em}{$\mathsf{C}$}\kern-0.2emT}}\xspace%
}

\newlength{\mylength}
\newcommand\setl[1]{\settowidth\mylength{#1}}

\newcommand\fixc[1]{\ \makebox[\mylength][c]{${#1}$}\ }

\begin{document}

\title[Modular Complexity Analysis]{Modular Complexity Analysis for
  Term Rewriting\rsuper*}

\author[H.~Zankl]{Harald Zankl\rsuper a}
\address{{\lsuper{a,b}}Institute of Computer Science\\University of
  Innsbruck\\Austria} 
\email{\{harald.zankl,martin.korp\}@uibk.ac.at}
\thanks{{\lsuper a}This research is supported by FWF (Austrian Science Fund)
  project P18763.}

\author[M.~Korp]{Martin Korp\rsuper b}
\address{\vspace{-18 pt}}

\keywords{term rewriting, complexity analysis, relative complexity,
derivation height}

\titlecomment{{\lsuper*}A preliminary version of this article appeared
  in RTA 2010.}

\begin{abstract}
\noindent
All current investigations to analyze the derivational complexity of term
rewrite systems are based on a single termination method, possibly
preceded by transformations. However, the exclusive use of direct
criteria is problematic due to their restricted power. To overcome this
limitation the article introduces a modular framework which allows to 
infer (polynomial) upper bounds on the complexity of term rewrite systems
by combining different criteria. Since the fundamental idea is based on 
relative rewriting, we study how matrix interpretations and match-bounds
can be used and extended to measure complexity for relative rewriting, 
respectively.
The modular framework is proved strictly more powerful than the conventional
setting. Furthermore, the results have been implemented and experiments
show significant gains in power.
\end{abstract}

\maketitle

\section{Introduction}

Term rewriting is a Turing complete model of computation. As an immediate
consequence all interesting properties are undecidable. Nevertheless many
powerful techniques have been developed to establish \emph{termination}.
The majority of these techniques have been automated successfully.
This development has been stimulated by the international competition of
termination tools.\footnote{\ \label{FOO:comp}\url{http://termcomp.uibk.ac.at}}
Most automated analyzers gain their power from a modular treatment
of rewrite systems (typically via the dependency pair
framework~\cite{AG00,HM05,T07}).

For terminating rewrite systems Hofbauer and Lautemann~\cite{HL89}
consider the length of derivations as a measurement for the complexity
of rewrite systems. The resulting notion of \emph{derivational complexity}
relates the length of a rewrite sequence to the size of its starting
term. Thereby it is, e.g., a suitable metric for the complexity of deciding
the word problem for a given confluent and terminating rewrite system
(since the decision procedure rewrites terms to normal form).
If one regards a rewrite system as a program and wants to estimate the
maximal number of computation steps needed to evaluate an expression to
a result, then the special shape of the starting terms---a function applied
to data which is in normal form---can be taken into account. Hirokawa and
Moser~\cite{HM08} identified this special form of complexity and named it
\emph{runtime complexity}.

To show (feasible) upper complexity bounds currently few techniques are
known. Typically termination criteria are restricted such that complexity
bounds can be inferred.
The early work by Hofbauer and Lautemann~\cite{HL89} considers polynomial
interpretations, suitably restricted, to admit quadratic derivational
complexity. Match-bounds~\cite{GHWZ07} and arctic matrix
interpretations~\cite{KW09} induce linear upper bounds on the
derivational complexity and triangular matrix interpretations~\cite{MSW08}
admit at most polynomially long derivations (the dimension of the matrices
yields the degree of the polynomial) in the size of the starting term.
All these methods share the property that until now they have been used
directly only, meaning that a single termination technique has to orient
all rules in one go. However, using direct criteria exclusively is
problematic due to their restricted power.

In~\cite{HM08,HM08b} Hirokawa and Moser lifted many aspects of the
dependency
pair framework from termination analysis into the complexity setting,
resulting in the notion of weak dependency pairs. So for the special
case of runtime complexity for the first time a modular approach
has been introduced. There the modular aspect amounts to using
different interpretation based criteria for (parts of the) weak dependency
graph and the
usable rules. However, still all rewrite rules considered must be oriented
strictly in one go and only restrictive criteria may be applied for the
usable rules. A further drawback of weak dependency pairs is that they
may only be used for bounding runtime complexity while there seems to be no
hope to generalize the method to derivational complexity.

In this article we present a different approach which admits a fully
modular treatment. The approach is general enough that it applies to
derivational complexity (and hence also to runtime complexity) and
basic enough that it allows to combine completely different complexity
criteria such as match-bounds and (triangular) matrix interpretations.
By the modular combination of different base methods also gains in power
are achieved. These gains come in two flavors. On the one hand our approach
allows to obtain lower complexity bounds for several rewrite systems where
bounds have already been established before and on the other hand we
found bounds for systems that could not be dealt with so far automatically.
More specifically, there are systems where the modular combination of
different criteria allows to establish an upper bound while any of the
involved methods cannot succeed on its own.

The remainder of the article is organized as follows. In
Section~\ref{PRE:main} preliminaries about term rewriting and complexity
analysis are fixed.
Afterwards, Section~\ref{REL:main} familiarizes the reader with the
concept of a suitable complexity measurement for relative rewriting.
Furthermore, it
formulates a modular framework for complexity analysis
based on relative complexity.
Criteria for measuring relative complexity via interpretations and
match-bounds are presented in Sections~\ref{MAT:main}
and~\ref{BOUNDS:main}, respectively.
In Section~\ref{ASS:main}
we show that the modular setting is strictly more powerful than the
conventional approach. Our results have been implemented in the complexity
prover \CAT. The technical details can be inferred from
Section~\ref{IMP:main}. Section~\ref{EXP:main} is devoted to demonstrate
the power of the modular treatment by means of an empirical evaluation.
Section~\ref{CON:main} concludes.

This article is a restructured and extended version of~\cite{ZK10}. It also
incorporates the results from the two notes~\cite{ZK10a,ZK10b} presented at
informal workshops. Furthermore results and presentation have been generalized
to address both derivational and runtime complexity.

\section{Preliminaries}
\label{PRE:main}

We assume familiarity with (relative) term rewriting~\cite{BN98,G90,TERESE}.
Let~\SIG{F} be a signature and~\VAR{V} a disjoint set of variables.
The set of terms over~\SIG{F} and~\VAR{V} is denoted by~$\FVTERMS$ and
the set of ground terms over~\SIG{F} by~$\FTERMS$.
We write $\Fun(t)$ for the set of function symbols occurring in a term~$t$.
The size of a term~$t$ is denoted~$|t|$ and $\|t\|$ computes the
number of occurrences of function symbols in~$t$.
A term~$t$ is called \emph{linear} if any variable~$x$ occurs at most
once in~$t$. 
Positions are used to address symbol occurrences in terms. Given a term $t$
we use $\Pos(t)$ to denote the set of positions induced by the term $t$ and
we write $t(p)$ with $p \in \Pos(t)$ for the symbol at position~$p$ in the term
$t$. The subset of positions $p \in \Pos(t)$ such that $t(p) \in \SIG{F}$
is denoted by $\FPos(t)$.

A \emph{rewrite rule} is a pair of terms $(l,r)$, written $l \to r$ such
that $l$ is not a variable and all variables in~$r$ are contained in~$l$.
A rewrite rule $l \to r$ is \emph{size-preserving} (\emph{size-decreasing})
if $|l| = |r|$ ($|l| > |r|$).
A \emph{term rewrite system} (TRS for short) is a set of rewrite rules.
For complexity analysis we assume TRSs to be finite and terminating.
A TRS $\TRS{R}$ is said to be \emph{duplicating} if there exist a rewrite
rule $l \to r \in \TRS{R}$ and a variable~$x$ that occurs more often in $r$
than in $l$. 
A TRS~\TRS{R} is called \emph{linear} (\emph{left-linear}, \emph{right-linear})
if for all rewrite rules $l \to r \in \TRS{R}$ the terms $l$ and $r$ ($l$, $r$)
are linear.
We call a TRS $\TRS{R}$ \emph{collapsing} if it
contains a rewrite rule $l \to r$ such that $r$ is a variable. The
\emph{defined symbols} of a TRS $\TRS{R}$ are all function symbols
$f$ for which there is a rewrite rule $l \to r$ in $\TRS{R}$ such that
$f = l(\epsilon)$. In the following we denote this set of function symbols
by $\DFun(\TRS{R})$. Those function symbols of $\TRS{R}$ which are not
defined are called \emph{constructor symbols}. So the set of all constructor
symbols is defined as $\CFun(\TRS{R}) = \SIG{F} \setminus \DFun(\TRS{R})$.

A \emph{rewrite relation} is a binary relation on terms that is
closed under contexts and substitutions. For a TRS~\TRS{R} we define
$\to_\TRS{R}$ to be the smallest rewrite relation that contains~\TRS{R}.
As usual~$\to^*$ denotes the reflexive and transitive closure of~$\to$
and $\to^m$ the $m$-th iterate of~$\to$.
A \emph{relative} TRS~$\REL{\TRS{R}}{\TRS{S}}$ is a pair of TRSs~\TRS{R}
and~\TRS{S} with the induced rewrite relation
${\to_\REL{\TRS{R}}{\TRS{S}}} =
{\to_{\TRS{S}}^* \cdot \to_{\TRS{R}}^{} \cdot \to_{\TRS{S}}^*}$. In the
sequel we will sometimes identify a TRS~\TRS{R} with the relative
TRS~$\REL{\TRS{R}}{\varnothing}$ and vice versa.
Furthermore properties defined for TRSs (as the ones above) naturally
extend to relative TRSs.

The \emph{derivation height} of a term $t$ with respect to a relation $\to$
is defined as follows: $\dl{t}{\to} = \sup\,\{m \mid \exists u\; t \to^m u\}$.
The \emph{complexity} of a relation $\to$ with respect to a (possibly infinite)
set of terms (or language)~$L$, denoted by $\cp{n}{\to}{L}$, computes the
maximal derivation height of all terms in~$L$ up to size~$n$ and is defined as
$\cp{n}{\to}{L} = \sup\,\{\dl{t}{\to} \mid
\text{$t \in L$ and $|t| \leqslant n$}\}$.
Sometimes we say that a TRS $\TRS{R}$ (relative TRS $\REL{\TRS{R}}{\TRS{S}}$)
has linear, quadratic, etc.\ or polynomial complexity with respect to~$L$ if
$\cp{n}{\to_\TRS{R}}{L}$ ($\cp{n}{\to_{\REL{\TRS{R}}{\TRS{S}}}}{L}$) can be
bounded by a linear, quadratic, etc.\ function or polynomial in~$n$.
Let $\TRS{R}$ be a TRS over some signature $\SIG{F}$. The
\emph{derivational complexity} of $\TRS{R}$, abbreviated by
$\dc{n}{\TRS{R}}$ and defined as
$\dc{n}{\TRS{R}} = \cp{n}{\to_\TRS{R}}{\FVTERMS}$,
computes the complexity of $\to_\TRS{R}$ with respect to all
terms. In contrast, the \emph{runtime complexity} of $\TRS{R}$
considers the maximal derivation height of constructor-based terms only,
i.e., $\rc{n}{\TRS{R}} = \cp{n}{\to_\TRS{R}}{\CGTERMS{\TRS{R},{\VAR{V}}}}$.
Here, the set of \emph{constructor-based terms} $\CGTERMS{\TRS{R},\VAR{V}}$
is defined as the set of all terms $t = f(t_1,\ldots,t_m)$ such that
$f \in \DFun(\TRS{R})$ and $t_i \in \GTERMS{\CFun(\TRS{R}),\VAR{V}}$
for all $i \in \{1,\dots,m\}$.

For functions $f,g \colon \Nat \to \Nat$ we write $f(n) = \OO(g(n))$ if there
are constants $M, N \in \Nat$ such that $f(n) \leqslant M\cdot g(n)$ for
all $n \geqslant N$. Furthermore, $f(n) = \Omega(g(n))$ if $g(n) =
\OO(f(n))$ and $f(n) = \Theta(g(n))$ if
$f(n) = \OO(g(n))$ and $f(n) = \Omega(g(n))$.

\section{Modular Complexity via Relative Complexity}
\label{REL:main}

In this section we present the basic idea that allows a modular treatment of
complexity proofs. To this end we introduce complexity analysis for relative
rewriting, i.e., given a relative TRS $\REL{\TRS{R}}{\TRS{S}}$ only the
\TRS{R}-steps contribute to the complexity. To estimate the derivational
complexity of a relative TRS $\REL{\TRS{R}}{\TRS{S}}$, a pair of orderings
$(\relgt,\relge)$ will be used such that $\TRS{R} \subseteq {\relgt}$ and
$\TRS{S} \subseteq {\relge}$. The necessary properties of these orderings
are given in the next definition.

\begin{defi}
A \emph{complexity pair} $(\relgt, \relge)$ consists of two finitely
branching rewrite relations~$\relgt$ and~$\relge$ that are
\emph{compatible}, i.e., ${\relge \cdot \relgt} \subseteq {\relgt}$ and
${\relgt \cdot \relge} \subseteq {\relgt}$.
We call a relative TRS $\REL{\TRS{R}}{\TRS{S}}$ \emph{compatible} with a
complexity pair $(\relgt,\relge)$ if $\TRS{R} \subseteq {\relgt}$
and $\TRS{S} \subseteq {\relge}$.
\end{defi}

The next lemma states that given a relative TRS $\REL{\TRS{R}}{\TRS{S}}$
and a compatible complexity pair $(\relgt,\relge)$, the $\relgt$
ordering is crucial for estimating the derivational complexity of
$\REL{\TRS{R}}{\TRS{S}}$. Intuitively the result states that every
$\REL{\TRS{R}}{\TRS{S}}$-step gives rise to at least one $\relgt$-step.

\begin{lem}
\label{LEM:bound}
Let $\REL{\TRS{R}}{\TRS{S}}$ be a relative TRS
compatible with a complexity pair $(\relgt,\relge)$.
Then for any term~$t$ we have
$\dl{t}{\relgt} \geqslant \dl{t}{\to_\REL{\TRS{R}}{\TRS{S}}}$.
\end{lem}
\begin{proof}
By assumption $\REL{\TRS{R}}{\TRS{S}}$ is compatible with $(\relgt,\relge)$.
Since $\relgt$ and $\relge$ are rewrite relations
${\to_\TRS{R}} \subseteq {\relgt}$ and ${\to_\TRS{S}} \subseteq {\relge}$
holds. From the compatibility of $\relgt$ and $\relge$ we obtain
${\to_{\REL{\TRS{R}}{\TRS{S}}}} \subseteq {\relgt}$.
Hence for any sequence
\setl{$\to_{\REL{\TRS{R}}{\TRS{S}}}$}
\begin{align*}
t\fixc{\to_{\REL{\TRS{R}}{\TRS{S}}}}
t_1\fixc{\to_{\REL{\TRS{R}}{\TRS{S}}}}
t_2\fixc{\to_{\REL{\TRS{R}}{\TRS{S}}}} \cdots
\end{align*}
also
\begin{align*}
t\fixc{\relgt}
t_1\fixc{\relgt}
t_2\fixc{\relgt} \cdots
\end{align*}
holds. The result follows immediately from this.
\end{proof}

Obviously $\relgt$ must be at least well-founded if \emph{finite} complexities
should be estimated. Because we are especially interested in feasible upper
bounds the following corollary is specialized to polynomials.

\begin{cor}
\label{COR:bound}
Let $\REL{\TRS{R}}{\TRS{S}}$ be a relative TRS compatible with a
complexity pair $(\relgt,\relge)$. If the complexity of~$\relgt$
with respect to some language~$L$ is linear, quadratic, etc.\
or polynomial then the complexity of~$\REL{\TRS{R}}{\TRS{S}}$ with
respect to~$L$ is linear, quadratic, etc. or polynomial.
\end{cor}
\begin{proof}
By Lemma~\ref{LEM:bound}.
\end{proof}

This corollary allows to investigate the complexity of (compatible)
complexity pairs instead of the complexity of the underlying relative
TRS. Sections~\ref{MAT:main} and~\ref{BOUNDS:main} are dedicated to
formulate powerful complexity pairs.
A severe drawback of complexity pairs is that given a relative
TRS~\REL{\TRS{R}}{\TRS{S}} all rules in \TRS{R} must be oriented
strictly. In the following we present a modular approach which allows
to combine different techniques for estimating the complexity of a relative
TRS~\REL{\TRS{R}}{\TRS{S}} with respect to a language~$L$. The fundamental
idea is based on the following simple procedure. Instead of computing the
complexity of \REL{\TRS{R}}{\TRS{S}} at once we try to bound the
complexity of \REL{\TRS{R}}{\TRS{S}} by splitting \TRS{R} into smaller
components $\TRS{R}_1$ and $\TRS{R}_2$. Here
$\TRS{R} = \TRS{R}_1 \cup \TRS{R}_2$. The aim is to over-estimate
$\dl{t}{\to_{\REL{\TRS{R}}{\TRS{S}}}}$ by
$\dl{t}{\to_{\REL{\TRS{R}_1}{(\TRS{R}_2 \cup \TRS{S})}}} +
\dl{t}{\to_{\REL{\TRS{R}_2}{(\TRS{R}_1 \cup \TRS{S})}}}$.
For each relative TRS $\REL{\TRS{R}_i}{(\TRS{R}_{3-i} \cup \TRS{S})}$
with $i \in \{1,2\}$ we can proceed in two directions: we can either
split up $\TRS{R}_i$ into smaller components or over-estimate
$\dl{t}{\to_{\REL{\TRS{R}_i}{(\TRS{R}_{3-i} \cup \TRS{S})}}}$ by applying
some suitable method. (Section~\ref{IMP:main} shows that this choice is
performed automatically.) Finally the complexity of the original system is
determined by summing up all intermediate results. The next
lemma states the main observation in this direction.

\begin{lem}
\label{LEM:mod}
Let~$\REL{(\TRS{R}_1 \cup \TRS{R}_2)}{\TRS{S}}$ be a relative TRS and
let~$t$ be a terminating term. Then
 $
\dl{t}{\to_\REL{\TRS{R}_1}{(\TRS{R}_2\cup\TRS{S})}} +
\dl{t}{\to_\REL{\TRS{R}_2}{(\TRS{R}_1\cup\TRS{S})}}
\geqslant
\dl{t}{\to_\REL{(\TRS{R}_1 \cup \TRS{R}_2)}{\TRS{S}}}
$.
\end{lem}
\begin{proof}
We abbreviate $\TRS{R}_1 \cup \TRS{R}_2$ by~\TRS{R}
and $\TRS{R}_{3-i} \cup \TRS{S}$ by $\TRS{S}_i$ for $i \in \{1,2\}$.
Assume that $\dl{t}{\to_\REL{\TRS{R}}{\TRS{S}}} = m$.
Then there exists a rewrite sequence
\begin{align}
\label{SEQ:old}
t       \to_\REL{\TRS{R}}{\TRS{S}}
t_1     \to_\REL{\TRS{R}}{\TRS{S}}
t_2     \to_\REL{\TRS{R}}{\TRS{S}}
\cdots  \to_\REL{\TRS{R}}{\TRS{S}}
t_{m-1} \to_\REL{\TRS{R}}{\TRS{S}}
t_m
\end{align}
of length~$m$. Next we investigate this sequence for every relative
TRS $\REL{\TRS{R}_i}{\TRS{S}_i}$ ($1 \leqslant i \leqslant 2$)
where $m_i$ overestimates how often rules from $\TRS{R}_i$ have been
applied in the original sequence.
Fix~$i$. If the sequence~\eqref{SEQ:old} does not contain an $\TRS{R}_i$
step then $t \to_{\TRS{S}_i}^m t_m$
and $m_i = 0$. In the other case there exists a maximal
(with respect to $m_i$) sequence
\begin{align}
\label{SEQ:new}
t             \to_\REL{\TRS{R}_i}{\TRS{S}_i}
t_{j_1}       \to_\REL{\TRS{R}_i}{\TRS{S}_i}
t_{j_2}       \to_\REL{\TRS{R}_i}{\TRS{S}_i}
\cdots        \to_\REL{\TRS{R}_i}{\TRS{S}_i}
t_{j_{m_i-1}} \to_\REL{\TRS{R}_i}{\TRS{S}_i}
t_m
\end{align}
where $1 \leqslant j_1 < j_2 < \dots < j_{m_i} = m$. Together with the
fact that every rewrite rule in $\TRS{R}$ is contained in $\TRS{R}_1$
or $\TRS{R}_2$ we have $m_1 + m_2 \geqslant m$. If $m_i = 0$ we obviously
have $\dl{t}{\to_\REL{\TRS{R}_i}{\TRS{S}_i}} \geqslant m_i$ and
if $t \to_\REL{\TRS{R}_i}{\TRS{S}_i}^{m_i} t_m$ with $m_i > 0$ we know
that $\dl{t}{\to_\REL{\TRS{R}_i}{\TRS{S}_i}} \geqslant m_i$ by the
choice of sequence~\eqref{SEQ:new}. (Note that in both cases it can
happen that $\dl{t}{\to_\REL{\TRS{R}_i}{\TRS{S}_i}} > m_i$ because
sequence~\eqref{SEQ:old} need not be maximal with respect to
$\to_\REL{\TRS{R}_i}{\TRS{S}_i}$.) Putting things together yields
\[
 \dl{t}{\to_\REL{\TRS{R}_1}{\TRS{S}_1}} +
 \dl{t}{\to_\REL{\TRS{R}_2}{\TRS{S}_2}}
\geqslant
m_1 + m_2 \geqslant m = \dl{t}{\to_\REL{\TRS{R}}{\TRS{S}}}
\]
which concludes the proof.
\end{proof}

As already indicated in the proof, the statement of the above lemma
does not hold for equality. This is illustrated by the following example.

\begin{exa}
\label{EXPL:mod}
Consider the relative TRS $\REL{\TRS{R}}{\TRS{S}}$ with
$\TRS{R}=\{\m{a} \to \m{b}, \m{a} \to \m{c}\}$ and
$\TRS{S}=\varnothing$.
We have $\m{a} \to_\REL{\TRS{R}}{\TRS{S}} \m{b}$ or
$\m{a} \to_\REL{\TRS{R}}{\TRS{S}} \m{c}$. Hence
$\dl{\m{a}}{\to_\REL{\TRS{R}}{\TRS{S}}} = 1$.
However, the sum of the derivation heights
$\dl{\m{a}}{\to_\REL{\{\m{a} \to \m{b}\}}{\{\m{a} \to \m{c}\}}}$ and
$\dl{\m{a}}{\to_\REL{\{\m{a} \to \m{c}\}}{\{\m{a} \to \m{b}\}}}$
is~$2$.
\end{exa}

Although for Lemma~\ref{LEM:mod} equality cannot be established the
next result states that for complexity analysis this does not matter.

\begin{thm}
\label{THM:modeq}
Let~$\REL{(\TRS{R}_1 \cup \TRS{R}_2)}{\TRS{S}}$ be a relative TRS
and $L$ be a set of terminating terms. Then
$
\cp{n}{\to_\REL{(\TRS{R}_1\cup\TRS{R}_2)}{\TRS{S}}}{L}
=
 \Theta(\cp{n}{\to_\REL{\TRS{R}_1}{(\TRS{R}_2\cup\TRS{S})}}{L} +
        \cp{n}{\to_\REL{\TRS{R}_2}{(\TRS{R}_1\cup\TRS{S})}}{L})
$.
\end{thm}
\begin{proof}
We have to show that there are constants $M,N$ and $M',N'$
such that for any term $t \in L$ the following two properties hold
(for $N$ and $N'$ choose 0, i.e., a term $t$ being a normal form):
\begin{itemize}
\item
$
\dl{t}{\to_\REL{(\TRS{R}_1\cup\TRS{R}_2)}{\TRS{S}}} \leqslant
 M\cdot(\dl{t}{\to_\REL{\TRS{R}_1}{(\TRS{R}_2\cup\TRS{S})}} +
        \dl{t}{\to_\REL{\TRS{R}_2}{(\TRS{R}_1\cup\TRS{S})}})
$
\item
$
M'\cdot\dl{t}{\to_\REL{(\TRS{R}_1\cup\TRS{R}_2)}{\TRS{S}}} \geqslant
\dl{t}{\to_\REL{\TRS{R}_1}{(\TRS{R}_2\cup\TRS{S})}} +
\dl{t}{\to_\REL{\TRS{R}_2}{(\TRS{R}_1\cup\TRS{S})}}
$
\end{itemize}
The result then follows from this. Lemma~\ref{LEM:mod} shows the first
property with $M = 1$. For the second property we reason as
follows. Let $i \in \{1,2\}$ and $\TRS{S}_i = \TRS{R}_{3-i} \cup \TRS{S}$.
Since $t \to_\REL{\TRS{R}_i}{\TRS{S}_i} t'$ implies
$t \to^+_\REL{(\TRS{R}_1\cup\TRS{R}_2)}{\TRS{S}} t'$
we obtain
$
\dl{t}{\to_\REL{(\TRS{R}_1\cup\TRS{R}_2)}{\TRS{S}}}
\geqslant
 \dl{t}{\to_\REL{\TRS{R}_i}{\TRS{S}_i}}
$.
The claim is shown by choosing $M' = 2$.
\end{proof}

Theorem~\ref{THM:modeq} allows to split a relative
TRS~$\REL{(\TRS{R}_1\cup\TRS{R}_2)}{\TRS{S}}$
into \emph{smaller} components
$\REL{\TRS{R}_1}{(\TRS{R}_2\cup\TRS{S})}$ and
$\REL{\TRS{R}_2}{(\TRS{R}_1\cup\TRS{S})}$
and evaluate the complexities of these
components (e.g., by different complexity pairs) independently. Note that
this approach is not restricted to relative rewriting. To estimate the
complexity of a (non-relative) TRS~\TRS{R} just consider the
relative TRS~$\REL{\TRS{R}}{\varnothing}$. The next example shows how
proofs in the modular framework look like. Section~\ref{IMP:main} gives
more details on proof trees.

\begin{exa}
\label{EX:reverse}
Proofs in the modular setting can be viewed as trees. We sketch such a proof
in Figure~\ref{FIG:sketch} using the TRS~$\TRS{R}$ consisting of the following
five rules:
\begin{alignat*}{2}
1\colon&& \m{rev}(x) &\to \m{rev}'(x,\m{nil}) \\
2\colon&& \m{rev}'(\m{nil},y) &\to y \\
3\colon&& \m{rev}'(\m{cons}(x,y),z) &\to \m{rev}'(y,\m{append}(\m{cons}(x,\m{nil}),z)) \\
4\colon&& \m{append}(\m{nil},y) &\to y \\
5\colon&& \m{append}(\m{cons}(x,y),z) &\to \m{cons}(x,\m{append}(y,z))
\end{alignat*}
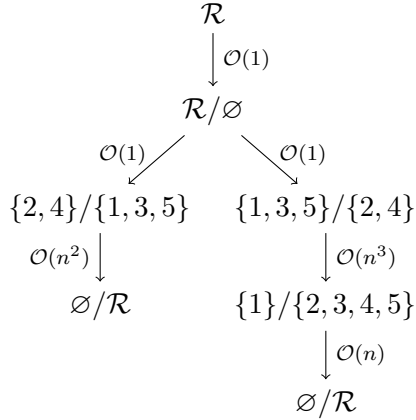
\begin{figure}
\begin{tikzpicture}[node distance=13mm and 15mm,on grid]
\node (1)                         {$\TRS{R}$};
\node (2)   [below=of 1]          {$\REL{\TRS{R}}{\varnothing}$};
\node (21)  [below left=of 2]     {$\REL{\{2,4\}}{\{1,3,5\}}$};
\node (211) [below=of 21]         {$\REL{\varnothing}{\TRS{R}}$};
\node (22)  [below right=of 2]    {$\REL{\{1,3,5\}}{\{2,4\}}$};
\node (221) [below=of 22]         {$\REL{\{1\}}{\{2,3,4,5\}}$};
\node (222) [below=of 221]        {$\REL{\varnothing}{\TRS{R}}$};
\draw[->] (1)   to node[right]      {{\scriptsize$\OO(1)$}}   (2);
\draw[->] (2)   to node[base left]  {{\scriptsize$\OO(1)$}}   (21);
\draw[->] (2)   to node[base right] {{\scriptsize$\OO(1)$}}   (22);
\draw[->] (21)  to node[left]       {{\scriptsize$\OO(n^2)$}} (211);
\draw[->] (22)  to node[right]      {{\scriptsize$\OO(n^3)$}} (221);
\draw[->] (221) to node[right]      {{\scriptsize$\OO(n)$}}   (222);
\end{tikzpicture}
\caption{Sketch of a modular complexity proof}
\label{FIG:sketch}
\end{figure}%
The root node of the tree is the TRS of interest and the other nodes are
relative rewrite systems representing intermediate
complexity problems. The edges indicate the (derivational) complexity of
the proof steps.  It is possible to apply Theorem~\ref{THM:modeq}
explicitly to split a problem into two (or more) problems as demonstrated
in the second node.  Such situations do not affect the complexity of the
given problem which justifies the labels~$\OO(1)$.  The remaining proof
steps measure the complexity of the rewrite rules that are moved from the
first into the second component (relative to the remaining rules).
These steps rely on an implicit application of Theorem~\ref{THM:modeq}.
For instance in the proof tree shown in Figure~\ref{FIG:sketch} there is an
edge from $\REL{\{1,3,5\}}{\{2,4\}}$ to $\REL{\{1\}}{\{2,3,4,5\}}$ labeled
$\OO(n^3)$, stating that the (derivational) complexity of
$\REL{\{3,5\}}{\{1,2,4\}}$ is at most cubic. This step is sound because
from Theorem~\ref{THM:modeq} we know that computing an upper bound on
$\REL{\{1\}}{\{2,3,4,5\}}$ and $\REL{\{3,5\}}{\{1,2,4\}}$ suffice to get
a valid upper bound on $\REL{\{1,3,5\}}{\{2,4\}}$.  In
Sections~\ref{MAT:main} and~\ref{BOUNDS:main} we study criteria that
allow to perform such proof steps.  Since the leaves in the tree give
rise to constant complexity, the complexity of the original problem can
be overestimated by summing up the complexities annotated to the edges;
yielding a cubic upper bound in this exemplary case. Later
(Example~\ref{EX:opt}) we will see that this bound is not tight.
\end{exa}

In the next two sections we study how matrix interpretations and the
match-bounds technique can be suited for relative complexity analysis.

\section{Matrix Interpretations}
\label{MAT:main}

This section is aimed at formulating complexity pairs based on matrix
interpretations~\cite{EWZ08}. Since our interest is in polynomial upper
bounds, triangular matrix interpretations~\cite{MSW08} and arctic matrix
interpretations~\cite{KW08} are considered. The last part of this section
generalizes the weight gap principle from~\cite{HM08} to 
(a restriction of) triangular matrix interpretations and relative rewriting.

\subsection{Preliminaries}

An \emph{\SIG{F}-algebra}~\ALG{A} consists of a non-empty carrier~$A$
and a set of interpretations~$f_\ALG{A}$ for every~$f\in\SIG{F}$. By
$[\alpha]_\ALG{A}(\cdot)$ we denote the usual evaluation function
of~\ALG{A} according to an assignment~$\alpha$.
An~\SIG{F}-algebra~\ALG{A} together with two relations~$\relgt$
and~$\relge$ on~$A$ is called a \emph{monotone algebra} if
every $f_\ALG{A}$ is monotone with respect to~$\relgt$ and~$\relge$,
$\relgt$ is a well-founded order, and $\relgt$ and $\relge$ are compatible.
Any monotone algebra~$(\ALG{A},\relgt,\relge)$ induces a well-founded
order on terms, i.e., $s \relgt_\ALG{A} t$ if for any assignment $\alpha$
the condition $[\alpha]_\ALG{A}(s) \relgt [\alpha]_\ALG{A}(t)$ holds. The
order $\relge_\ALG{A}$ is defined similarly.
A relative TRS~$\REL{\TRS{R}}{\TRS{S}}$ is \emph{compatible} with a
monotone algebra~$(\ALG{A},\relgt,\relge)$ if \REL{\TRS{R}}{\TRS{S}}
is compatible with $(\relgt_\ALG{A},\relge_\ALG{A})$.

\subsection{Triangular Matrix Interpretations}

\emph{Matrix interpretations}~$(\ALG{M},\relgt,\relge)$
(often just denoted~\ALG{M}) are a special form of monotone algebras.
Here the carrier is $\Nat^{d}$ for some fixed
dimension~$d \in \Nat \setminus \{0\}$. The order
$\relge$ is the point-wise extension of $\geqslant_\Nat$ to
vectors and
$\vec{u} \relgt \vec{v}$ if $u_1 >_\Nat v_1$ and
$\vec{u} \relge \vec{v}$. If every $f\in\SIG{F}$ of arity~$n$
is interpreted as
$f_\ALG{M}(\vec{x_1},\dots,\vec{x_n}) = F_1\vec{x_1} + \dots +
F_n\vec{x_n} + \vec{f}$ where $F_i \in \Nat^{d\times d}$ for all
$1\leqslant i \leqslant n$ and $\vec{f} \in \Nat^d$ then
monotonicity of~$\relgt$ is achieved by demanding
${F_i}_{(1,1)} \geqslant 1$ for any 
$f \in \SIG{F}$ and 
$1\leqslant i\leqslant n$.
Such interpretations have been introduced in~\cite{EWZ08}.

A matrix interpretation where for every $f\in\SIG{F}$ all $F_i$
($1\leqslant i\leqslant n$ where~$n$ is the arity of~$f$)
are upper triangular is called
\emph{triangular matrix interpretation} (abbreviated by TMI).
A square matrix~$A$ of dimension~$d$ is of \emph{upper triangular}
shape if $A_{(i,i)} \leqslant 1$ and $A_{(i,j)} = 0$ if $i > j$ for all
$1\leqslant i,j\leqslant d$.
For historic reasons a TMI based on matrices of dimension one is also
called \emph{strongly linear interpretation} (SLI for short).
In~\cite{MSW08} it is shown that the derivational complexity of
a TRS~\TRS{R} is bounded by a polynomial of degree~$d$ if there
exists a TMI~\ALG{M} of dimension$~d$
such that~$\TRS{R} \subseteq {\relgt_\ALG{M}}$.
For our setting the following formulation is more useful.

\begin{thm}
\label{THM:tmi}
Let $\ALG{M}$ be a TMI of dimension~$d$ over a
signature $\SIG{F}$. Then $(\relgt_\ALG{M},\relge_\ALG{M})$
is a complexity pair. Furthermore
$\cp{n}{\relgt_\ALG{M}}{\FVTERMS} = \OO(n^d)$.
\end{thm}
\begin{proof}
Straightforward from \cite[Theorem~6]{MSW08}.
\end{proof}

The following example familiarizes the reader with
TMIs.

\begin{exa}
\label{EX:tmi}
Consider the relative TRS~$\REL{\TRS{R}}{\TRS{S}}$ over the signature
$\SIG{F} = \{\m{f},\m{g}\}$ defined as
$\TRS{R} = \{\m{f}(\m{f}(x)) \to \m{f}(\m{g}(\m{f}(x)))\}$ and
$\TRS{S} = \{\m{f}(x) \to x\}$. Then the TMI~$\ALG{M}$ of dimension two
with
\begin{xalignat*}{2}
\m{f}_\ALG{M}(\vec x) &= \BM
1\NE 0\NR
0\NE 1\NR
\EM \vec x +
\BM
1\NR
1\NR
\EM
&
\m{g}_\ALG{M}(\vec x) &= \BM
1\NE 0\NR
0\NE 0\NR
\EM \vec x
\end{xalignat*}
induces the complexity pair
$\!(\relgt_\ALG{M},\relge_\ALG{M})$.
Furthermore $\!\REL{\TRS{R}}{\TRS{S}}\!$ is compatible
with $\!(\relgt_\ALG{M},\relge_\ALG{M})\!$.
Theorem~\ref{THM:tmi} gives a quadratic upper bound on
$\cp{n}{\relgt_\ALG{M}}{\FVTERMS}$.
Hence
the derivational complexity of~$\REL{\TRS{R}}{\TRS{S}}$ is at most quadratic
by Corollary~\ref{COR:bound}.
It is easy to see (cf.\ Example~\ref{EX:ami}) that this bound is not
tight. We remark that there cannot exist an SLI that establishes a linear
upper bound because no SLI can orient the rule $\m{f}(\m{f}(x)) \to
\m{f}(\m{g}(\m{f}(x)))$ strictly.
\end{exa}

\subsection{Arctic Matrix Interpretations}

We define $\Arctic = \Nat \cup \{\minfty\}$.
For matrices $A \in \Arctic^{n\times m}$ and
$B \in \Arctic^{m\times p}$ the operation
$\otimes$ yields an $n \times p$ matrix and is defined as follows:
$(A \otimes B)_{(i,j)} = \max_{1\leqslant k\leqslant m} \{A_{(i,k)}+B_{(k,j)}\}$
where $+$ and $\max$ are extended naturally to deal with~$\minfty$
(see~\cite{KW08}). Furthermore $x >_\Arctic y$ if and only if $x >_\Nat y$
or $x = y = \minfty$, and $x \geqslant_\Arctic y$ if and only if
$x \geqslant_\Nat y$ or $y = \minfty$.%
\footnote{\ Note that $\minfty >_\Arctic \minfty$ and hence $>_\Arctic$ is 
not well-founded. Hence such comparisons are disallowed at certain matrix
positions.}

An \emph{arctic matrix interpretation}~$(\ALG{A},\relgt,\relge)$
(abbreviated by AMI and often just denoted~\ALG{A}) is a special
form of a monotone algebra. 
Here the carrier is $\Arctic^{d}$ for some
fixed dimension~$d \in \Nat \setminus \{0\}$. The orders
$\relge$ and $\relgt$ are the
point-wise extensions of $\geqslant_\Arctic$ and $>_\Arctic$
to vectors, respectively.
Every unary function symbol~$f\in \SIG{F}$ is interpreted
as $f_\ALG{A}(\vec x) = F \otimes \vec x$ where
$F \in \Arctic^{d\times d}$ and every constant~$c$ as
$c_\ALG{A} = \vec c$ where $\vec c \in \Arctic^d$.
Monotonicity of~$\relgt$ is achieved by the restriction to at
most unary function symbols and by demanding that
$F_{(1,1)}$ and $c_1$ are different from~$\minfty$ for unary
function symbols~$f$ and constants~$c$, respectively.
In~\cite{KW08} it is shown that the derivational complexity of a
TRS~\TRS{R}, which contains unary and constant function symbols only, is
at most linear if there exists an AMI~\ALG{A} (of some dimension~$d$)
such that~$\TRS{R} \subseteq {\relgt_\ALG{A}}$.

\begin{thm}
\label{THM:ami}
Let $\ALG{A}$ be an AMI of dimension~$d$ over a signature $\SIG{F}$
that contains constants and unary function symbols only.
Then $(\relgt_\ALG{A},\relge_\ALG{A})$
is a complexity pair. Furthermore
$\cp{n}{\relgt_\ALG{A}}{\FVTERMS} = \OO(n)$.
\end{thm}
\begin{proof}
Straightforward from \cite[Lemma~17]{KW08}.
\end{proof}

\begin{exa}
\label{EX:ami}
Consider the TRSs from Example~\ref{EX:tmi}.
Then the AMI~$\ALG{A}$ satisfying
\begin{xalignat*}{2}
\m{f}_\ALG{A}(\vec x) &= \BM
1\NE 3\NR
0\NE 3\NR
\EM \vec x
&
\m{g}_\ALG{A}(\vec x) &= \BM
0\NE 1\NR
\minfty\NE \minfty\NR
\EM \vec x
\end{xalignat*}
induces the complexity pair
$\!(\relgt_\ALG{A},\relge_\ALG{A})$. Furthermore
$\!\REL{\TRS{R}}{\TRS{S}}$ is compatible with
$\!(\relgt_\ALG{A},\relge_\ALG{A})$.\break  Theorem~\ref{THM:ami} gives a
linear upper bound on $\cp{n}{\relgt_\ALG{A}}{\FVTERMS}$. Hence the
derivational complexity of~$\REL{\TRS{R}}{\TRS{S}}$ is at most linear
by Corollary~\ref{COR:bound}.  It is easy to see that this bound is
tight.
\end{exa}

\subsection{Complexity Gap Principle}

An obvious question is whether it suffices to estimate polynomial
complexity of $\REL{(\TRS{R}_1\cup\TRS{R}_2)}{\TRS{S}}$ by establishing
polynomial upper bounds on the complexities of
$\REL{\TRS{R}_1}{(\TRS{R}_2 \cup \TRS{S})}$ and $\REL{\TRS{R}_2}{\TRS{S}}$
(in contrast to $\REL{\TRS{R}_2}{(\TRS{R}_1 \cup \TRS{S})}$
as in Theorem~\ref{THM:modeq}). The following example by
Hofbauer~\cite{HW06t} shows that in general the complexity of
$\REL{(\TRS{R}_1 \cup \TRS{R}_2)}{\TRS{S}}$ might be much larger
than the sum of the components above; even for systems where both
parts have linear complexity. Here $\TRS{S} = \varnothing$.

\begin{exa}
\label{EX:counterexample}
Consider the TRS~$\TRS{R}_1$ consisting of the single rule
\begin{xalignat*}{3}
&& \m{c}(\m{L}(x)) &\to \m{R}(x) &&
\intertext{
and the TRS~$\TRS{R}_2$ consisting of the rewrite rules
}
\m{R}(\m{a}(x)) &\to \m{b}(\m{b}(\m{R}(x))) &
\m{R}(x) &\to \m{L}(x) &
\m{b}(\m{L}(x)) &\to \m{L}(\m{a}(x))
\end{xalignat*}
The derivational complexity of the relative
TRS $\REL{\TRS{R}_1}{\TRS{R}_2}$ is
linear, due to the SLI that just counts the $\m{c}$'s. The derivational
complexity of $\TRS{R}_2$ is linear as well since the system can be proved
terminating by the match-bound technique~\cite{GHWZ07}. However, the
TRS $\TRS{R}_1 \cup \TRS{R}_2$ admits exponentially long
derivations in the size of the starting term:
\[
\begin{array}{r@{~\to~}l@{~\to~}l@{~}l@{~}l}
\m{c}^n(\m{L}(\m{a}(x))) &
\m{c}^{n-1}(\m{R}(\m{a}(x))) &
\m{c}^{n-1}(\m{b}(\m{b}(\m{R}(x)))) & \to &
\m{c}^{n-1}(\m{b}(\m{b}(\m{L}(x)))) \\[1ex]
&
\m{c}^{n-1}(\m{b}(\m{L}(\m{a}(x)))) &
\m{c}^{n-1}(\m{L}(\m{a}(\m{a}(x)))) &\to^* &
\m{L}(\m{a}^{2^n}(x))
\end{array}
\]
\end{exa}

Under certain circumstances the problem of the preceding example
does not occur. Inspired by the weight gap principle of Hirokawa and
Moser~\cite{HM08} (which was developed to estimate weak dependency pair
steps relative to usable rule steps), below we state abstract criteria
on $\TRS{R}_1$ and $\TRS{R}_2$ such that the complexity of
$\REL{\TRS{R}_1}{(\TRS{R}_2 \cup \TRS{S})}$ and $\REL{\TRS{R}_2}{\TRS{S}}$
determines the complexity of $\REL{(\TRS{R}_1 \cup \TRS{R}_2)}{\TRS{S}}$.

\begin{thm}[Complexity Gap Principle]
\label{THM:dgp}
Let $\REL{(\TRS{R}_1 \cup \TRS{R}_2)}{\TRS{S}}$ be a relative TRS
and $L$ be a set of terminating terms. If there exist a complexity
pair $(\relgt,\relge)$ and a constant~$\Delta$ such that
$\REL{\TRS{R}_2}{\TRS{S}}$ is compatible with $(\relgt,\relge)$
and $u \to_{\TRS{R}_1} v$ implies
$\dl{u}{\relgt} + \Delta \geqslant \dl{v}{\relgt}$ then
$\cp{n}{\to_{\REL{(\TRS{R}_1 \cup \TRS{R}_2)}{\TRS{S}}}}{L} =
\OO(\cp{n}{\to_{\REL{\TRS{R}_1}{(\TRS{R}_2 \cup \TRS{S})}}}{L} +
\cp{n}{\relgt}{L})$.
\end{thm}
\begin{proof}
We show that under the above assumptions, for any term $s \in L$
there exists a constant~$M$ such that
$\dl{s}{\to_{\REL{(\TRS{R}_1 \cup \TRS{R}_2)}{\TRS{S}}}} \leqslant
M \cdot \dl{s}{\to_{\REL{\TRS{R}_1}{(\TRS{R}_2 \cup \TRS{S})}}} +
\dl{s}{\relgt}$.
Consider a derivation of maximal length
in $\REL{(\TRS{R}_1 \cup \TRS{R}_2)}{\TRS{S}}$,
written as follows:
\begin{xalignat}{1}
\label{CGP:seq}
s = s_0 \to^{k_0}_{\REL{\TRS{R}_2}{\TRS{S}}} \cdot \to^*_\TRS{S}
t_0 \to_{\TRS{R}_1}
s_1 \to^{k_1}_{\REL{\TRS{R}_2}{\TRS{S}}} \cdot \to^*_\TRS{S}
t_1 \to_{\TRS{R}_1}
\cdots \to_{\TRS{R}_1}
s_m \to^{k_m}_{\REL{\TRS{R}_2}{\TRS{S}}} \cdot \to^*_\TRS{S}
t_m
\end{xalignat}
Since sequence~\eqref{CGP:seq} is maximal,
$\dl{s_0}{\to_{\REL{(\TRS{R}_1\cup \TRS{R}_2)}{\TRS{S}}}} \leqslant
\dl{s_0}{\to_{\REL{\TRS{R}_1}{(\TRS{R}_2 \cup \TRS{S})}}} +
\sum_{0\leqslant i\leqslant m} k_i$.
Because $\REL{\TRS{R}_2}{\TRS{S}}$ is compatible with $(\relgt,\relge)$
we have $\dl{s_0}{\relgt} \geqslant \dl{t_0}{\relgt} + k_0$.
From the assumption, $\dl{t_0}{\relgt} + \Delta \geqslant \dl{s_1}{\relgt}$
follows and hence
$\dl{s_0}{\relgt} + \Delta \geqslant \dl{s_1}{\relgt} + k_0$.
Repeating this argument shows
$\dl{s_0}{\relgt} + {m\cdot\Delta} \geqslant
\sum_{0\leqslant i\leqslant m} k_i$.
Because $m \leqslant \dl{s_0}{\to_{\REL{\TRS{R}_1}{(\TRS{R}_2 \cup \TRS{S})}}}$
(note that equality does not hold since sequence~\eqref{CGP:seq} need not be
maximal for $\REL{\TRS{R}_1}{(\TRS{R}_2 \cup \TRS{S})}$) we obtain
$\dl{s_0}{\to_{\REL{(\TRS{R}_1 \cup \TRS{R}_2)}{\TRS{S}}}} \leqslant
\dl{s_0}{\to_{\REL{\TRS{R}_1}{(\TRS{R}_2 \cup \TRS{S})}}} + \dl{s_0}{\relgt} +
{\dl{s_0}{\to_{\REL{\TRS{R}_1}{(\TRS{R}_2 \cup \TRS{S})}}}} \cdot \Delta$
which simplifies to
$\dl{s_0}{\to_{\REL{(\TRS{R}_1\cup \TRS{R}_2)}{\TRS{S}}}} \leqslant
(\Delta + 1) \cdot \dl{s_0}{\to_{\REL{\TRS{R}_1}{(\TRS{R}_2 \cup \TRS{S})}}} +
\dl{s_0}{\relgt}$.
Finally, taking $M = \Delta + 1$ concludes the proof.
\end{proof}

To implement the above theorem the question arises which further requirements
besides compatibility of $\REL{\TRS{R}_2}{\TRS{S}}$ with a complexity pair
$(\relgt,\relge)$ are required such that for any terms $u$ and $v$ a step
$u \to_{\TRS{R}_1} v$ implies the desired
$\dl{u}{\relgt} + \Delta \geqslant \dl{s}{\relgt}$ for some constant~$\Delta$.
One idea is to test $\dl{l}{\relgt} + \Delta \geqslant \dl{r}{\relgt}$
explicitly 
for any $l\to r \in \TRS{R}_1$ and 
demand that the complexity pair $(\relgt,\relge)$ then satisfies
$\dh{C[l\sigma]}{\relgt} + \Delta \geqslant \dh{C[r\sigma]}{\relgt}$
for all contexts~$C$ and substitutions~$\sigma$.

As we know from~\cite{HM08}, SLIs can be used to get a concrete instance
of Theorem~\ref{THM:dgp} with respect to derivational complexity, if
\TRS{S} is empty. Below we state the result in the relative setting, which
is more useful for our purposes.

\begin{cor}
\label{COR:sli}
Let $\REL{(\TRS{R}_1 \cup \TRS{R}_2)}{\TRS{S}}$ be a relative TRS,
$\TRS{R}_1$ be non-duplicating, and $\REL{\TRS{R}_2}{\TRS{S}}$ be
compatible with an SLI. Then
$\dc{n}{\REL{(\TRS{R}_1 \cup \TRS{R}_2)}{\TRS{S}}} =
\OO(\dc{n}{\REL{\TRS{R}_1}{(\TRS{R}_2 \cup \TRS{S})}} + n)$.
\end{cor}
\begin{proof}
Follows from Theorems~\ref{THM:dgp} and~\ref{THM:tmi} using the complexity pair
$(\relgt_\ALG{M},\relge_\ALG{M})$ induced by the SLI~$\ALG{M}$.
\end{proof}

An immediate consequence of the above corollary is that for any relative
TRS~$\REL{\TRS{R}}{\TRS{S}}$ we can shift rewrite rules in~\TRS{R} that
are strictly oriented by an SLI~\ALG{M} into the~\TRS{S}-component,
provided that~\TRS{R} is non-duplicating and all rules in~\TRS{S} behave
nicely with respect to $\relge_\ALG{M}$. Note that the above corollary
does not require that all rules from~\TRS{R} are (strictly) oriented.
This causes some kind of non-determinism which is demonstrated in the
next example.

\begin{exa}
\label{EX:non-determinism}
Consider the TRS (\tpdb{Bouchare\_06/12})\footnote{\ Labels in
$\mathsf{sans\text{-}serif}$ font refer to TRSs from the TPDB~7.0.2,
see \url{http://termination-portal.org}.} consisting of the rules:
\begin{xalignat*}{3}
1\colon\m{b}(\m{b}(x) &\to \m{a}(\m{a}(\m{a}(x)))&
2\colon\m{b}(\m{a}(\m{b}(x))) &\to \m{a}(x)&
3\colon\m{b}(\m{a}(\m{a}(x))) &\to \m{b}(\m{a}(\m{b}(x)))
\end{xalignat*}
The SLI~$\ALG{M}$ with
$\m{a}_\ALG{M}(x) = x + 2$ and
$\m{b}_\ALG{M}(x) = x + 1$ transforms the TRS into
$\REL{\{1\}}{\{2,3\}}$ which is compatible with the
AMI~\ALG{A} (where all matrix coefficients are smaller than two)
\begin{xalignat*}{2}
\m{a}_\ALG{A}(\vec x) &= \BM
0 \NE 0 \NE 0 \NR
\minfty \NE \minfty \NE 0 \NR
\minfty \NE 0 \NE \minfty \NR
\EM \vec x
&
\m{b}_\ALG{A}(\vec x) &= \BM
0 \NE 1 \NE 0 \NR
0 \NE 1 \NE 0 \NR
\minfty \NE 0 \NE \minfty \NR
\EM \vec x
\end{xalignat*}
showing linear derivational complexity of this TRS. If a different SLI is
used in the first step, e.g., the one that counts just $\m{b}$'s then the
intermediate problem $\REL{\{3\}}{\{1,2\}}$ remains to be solved. For
this problem there exists no AMI of dimension three where all
entries are less than 2 (but there exists one where all entries are less
than 3). For an implementation this means that depending on
the rules the SLI orients, later techniques may succeed or fail.
\end{exa}

Next we remark on another subtlety of~Theorem~\ref{THM:dgp}.
Assume that $\REL{\TRS{R}_1}{(\TRS{R}_2\cup\TRS{S})}$ is compatible
with a complexity pair $(\relgt,\relge)$. Then
$\REL{(\TRS{R}_1\cup\TRS{R}_2)}{\TRS{S}}$
is transformed into the problem
$\REL{\TRS{R}_2}{(\TRS{R}_1\cup\TRS{S})}$ and this proof step estimates
the complexity of $\REL{\TRS{R}_1}{(\TRS{R}_2\cup\TRS{S})}$.
If the complexity gap principle is used the situation changes.
Since it does not require (weak) compatibility with $\TRS{R}_1$, it does not
make a statement about the complexity of
$\REL{\TRS{R}_1}{(\TRS{R}_2\cup\TRS{S})}$. Instead it states that
the complexity of $\REL{(\TRS{R}_1\cup\TRS{R}_2)}{\TRS{S}}$ is
dominated by the complexity of $\REL{\TRS{R}_1}{(\TRS{R}_2\cup\TRS{S})}$
or the complexity of $\REL{\TRS{R}_2}{\TRS{S}}$.
This behavior is illustrated in the next example.

\begin{exa}
\label{EX:wgp}
Consider the relative TRS
$\TRS{R}$ consisting of the two rules
\begin{xalignat*}{2}
1: \m{c}(x) &\to \m{a}(x) &
2: \m{a}(\m{b}(x)) &\to \m{b}(\m{b}(\m{c}(x)))
\end{xalignat*}
We observe that the derivational complexity of the TRS~$\TRS{R}$ is at
least exponential because
\begin{align*}
\m{a}^n(\m{b}(x)) \to^2
\m{a}^{n-1}(\m{b}(\m{b}(\m{a}(x)))) \to^4
\m{a}^{n-2}(\m{b}(\m{b}(\m{b}(\m{b}(\m{a}(x)))))) \to^8
\cdots \to^{2^n}
\m{b}^{2^{n}}(\m{a}(x))
\end{align*}
Obviously both rules are applied exponentially often in this sequence.
Nevertheless by an SLI that counts $\m{c}$'s Corollary~\ref{COR:sli}
can be applied to~$\TRS{R}$ to obtain the relative TRS~$\REL{\{2\}}{\{1\}}$.
As remarked earlier this step does not yield an upper
bound on the complexity of the TRS~$\REL{\{1\}}{\{2\}}$ but only
on the TRS~$\{1\}$.
\end{exa}

Next we give counterexamples that TMIs, AMIs, and
match-bounds cannot be used to implement Theorem~\ref{THM:dgp}.
A suitable but severe restriction of TMIs is considered in
Theorem~\ref{THM:cgp_tmi}.

\begin{exas}
Recall the two TRSs $\TRS{R}_1$ and $\TRS{R}_2$ from
Example~\ref{EX:counterexample}. Here $\TRS{S} = \varnothing$.
Since $\dc{n}{\REL{\TRS{R}_1}{\TRS{R}_2}} = \OO(n)$ and
$\dc{n}{\TRS{R}_1\cup\TRS{R}_2} = \Omega(2^n)$ any method
that establishes $\dc{n}{\TRS{R}_2} = \OO(n^k)$ for some
$k \in \Nat$ cannot be used to implement the complexity gap principle.
\item
Since the TMI~$\ALG{M}$ with
\begin{xalignat*}{4}
\m{a}_\ALG{M}(\vec x) &= \vec x + \BM 0\NR1 \EM &
\m{b}_\ALG{M}(\vec x) &= \BM 1\NE0\NR0\NE0 \EM \vec x + \BM 1\NR0 \EM &
\m{R}_\ALG{M}(\vec x) &= \BM 1\NE3\NR0\NE0 \EM \vec x + \BM 2\NR0 \EM &
\m{L}_\ALG{M}(\vec x) &= \BM 1\NE0\NR0\NE0 \EM \vec x
\end{xalignat*}
orients all rules in $\TRS{R}_2$ strictly---and hence gives a quadratic
upper bound on $\dc{n}{\TRS{R}_2}$---in general TMIs cannot adhere to
Theorem~\ref{THM:dgp}. The problem for the interpretation above is
that although there exists
a~$\Delta$ with $\dl{l}{\relgt} + \Delta \geqslant \dl{r}{\relgt}$
for all $l \to r \in \TRS{R}_1$ this property is not closed under
substitutions. (The situation is different, however, if the matrix
interpretation has constant growth, see Theorem~\ref{THM:cgp_tmi} below.)
\item
Similarly, the AMI~$\ALG{A}$ (inducing at most linear derivational
complexity of~$\TRS{R}_2$) with
\begin{xalignat*}{4}
\m{a}_\ALG{A}(\vec x) &= \BM 0\NE\minfty\NR3\NE3 \EM \vec x &
\m{b}_\ALG{A}(\vec x) &= \BM 1\NE2\NR\minfty\NE0 \EM \vec x &
\m{R}_\ALG{A}(\vec x) &= \BM 1\NE3\NR0\NE2 \EM \vec x &
\m{L}_\ALG{A}(\vec x) &= \BM 0\NE\minfty\NR\minfty\NE\minfty \EM \vec x
\end{xalignat*}
violates the same requirement in Theorem~\ref{THM:dgp} as the TMI
$\ALG{M}$ above.
\item
A similar reasoning also holds for match-bounds; one easily verifies that
match-bounds apply to the TRS~$\TRS{R}_2$ and hence this system admits
linear derivational complexity. The problem in this setting is that a valid
termination proof of $\TRS{R}_2$ using match-bounds does not necessarily
yield a rewrite relation $\relgt$ such that
$\dl{u}{\relgt} + \Delta \geqslant \dl{v}{\relgt}$
whenever $u \to_{\TRS{R}_1} v$, as required by Theorem~\ref{THM:dgp}.
\end{exas}

Finally we present a criterion that allows to implement Theorem~\ref{THM:dgp}
based on TMIs. To this end we introduce the following concepts. 
A matrix interpretation~\ALG{M} has \emph{constant growth} if there is a
matrix $A$ such that for any $p \in \Nat$ and matrices $M_1, \ldots, M_p$
in~\ALG{M} we have $M_1 \cdot \ldots \cdot M_p \leqslant A$. Here
$\leqslant$ is the pointwise extension of $\leqslant_\Nat$ to matrices.
Because of the shape of matrix interpretations for terms $s$ and $t$
there exist $k \in \Nat$, matrices $S_1, \ldots, S_k, T_1, \ldots, T_k$, 
and vectors $\vec s$, $\vec t$ such that
$[\alpha]_\ALG{M}(s) = S_1\alpha(x_1) + \cdots + S_k\alpha(x_k) + \vec s$
and
$[\alpha]_\ALG{M}(t) = T_1\alpha(x_1) + \cdots + T_k\alpha(x_k) + \vec t$.
In such a case we denote the non-constant part of the interpretation of $s$
by
$[\alpha]^\ncp_\ALG{M}(s) = S_1\alpha(x_1) + \cdots + S_k\alpha(x_k)$;
similarly for $t$.
We write $s \relge_\ALG{M}^\ncp t$ if 
$[\alpha]^\ncp_\ALG{M}(s) \relge [\alpha]^\ncp_\ALG{M}(t)$ holds for all
assignments $\alpha$. Note that this condition can effectively be tested
by requiring $S_i \geqslant T_i$ ($1\leqslant i\leqslant k$). 

\begin{thm}
\label{THM:cgp_tmi}
Let $\REL{(\TRS{R}_1 \cup \TRS{R}_2)}{\TRS{S}}$ be a relative TRS,
$L$ a set of terminating terms,
$\ALG{M}$ a matrix interpretation with constant growth,
$\TRS{R}_1 \subseteq {\relge^\ncp_\ALG{M}}$, and $\REL{\TRS{R}_2}{\TRS{S}}$ be
compatible with~\ALG{M}. Then
$\cp{n}{\to_\REL{(\TRS{R}_1 \cup \TRS{R}_2)}{\TRS{S}}}{L} =
\OO(\cp{n}{\to_\REL{\TRS{R}_1}{(\TRS{R}_2 \cup \TRS{S})}}{L} + n)$.
\end{thm}
\begin{proof}
Throughout this proof we assume that $L = \FVTERMS$.
Since the matrix interpretation~\ALG{M} has constant growth we have
$\cp{n}{\relgt_\ALG{M}}{L} = \OO(n)$.
Since $\REL{\TRS{R}_2}{\TRS{S}}$ is compatible with the complexity pair
$(\relgt_\ALG{M},\relge_\ALG{M})$ 
using Theorem~\ref{THM:dgp} it remains to show that there is a constant
$\Delta$ such that $u \to_{\TRS{R}_1} v$ implies 
$\dh{u}{\relgt_\ALG{M}} + \Delta \geqslant \dh{v}{\relgt_\ALG{M}}$.
Since~\ALG{M} has constant growth there is a matrix $A$ such that
$A \geqslant M_1\cdot \ldots \cdot M_p$ for any $p\in\Nat$ where the
$M_i$'s are matrices occurring in~\ALG{M}. 
Let $\delta = \max\,\{\vec r \mid l \to r \in \TRS{R}_1\}$ (here $\vec r$
is the constant part of the interpretation of $r$ and $\max$
denotes the pointwise maximum of vectors). Note that $\delta$ is a
vector.

Let $\Delta = (A\delta)_{11}$. Because the derivation height of
a term~$t$ with respect to
$\relgt_\ALG{M}$ is determined by the first component of the vector 
$[\alpha]_\ALG{M}(t)$
we have
$\dh{u}{\relgt_\ALG{M}} + \Delta \geqslant \dh{v}{\relgt_\ALG{M}}$
whenever
$[\alpha]_\ALG{M}(u) + A\delta \geqslant [\alpha]_\ALG{M}(v)$.
To show the latter let 
$l\to r \in \TRS{R}_1$,
$u = C[l\sigma]$,
$v = C[r\sigma]$,
$[\alpha]_\ALG{M}(l) = L_1\alpha(x_1) + \cdots + L_k\alpha(x_k) + \vec l$,
and
$[\alpha]_\ALG{M}(r) = R_1\alpha(x_1) + \cdots + R_k\alpha(x_k) + \vec r$.

By definition of $\delta$ we have $\vec l + \delta \geqslant \vec r$.
Since $\TRS{R}_1 \subseteq {\relge^\ncp_\ALG{M}}$ we have 
$L_i \geqslant R_i$ for all $1\leqslant i\leqslant k$ and hence
$L_1\alpha(x_1) + \cdots + L_k\alpha(x_k) + \vec l + \delta \geqslant 
 R_1\alpha(x_1) + \cdots + R_k\alpha(x_k) + \vec r$
for any $\alpha$ and furthermore 
$L_1\alpha(x_1\sigma) + \cdots + L_k\alpha(x_k\sigma) + \vec l + \delta 
 \geqslant 
 R_1\alpha(x_1\sigma) + \cdots + R_k\alpha(x_k\sigma) + \vec r$ for any
$\sigma$. The latter implies
$D(L_1\alpha(x_1\sigma) + \cdots + L_k\alpha(x_k\sigma) + \vec l + \delta) 
\geqslant 
D(R_1\alpha(x_1\sigma) + \cdots + R_k\alpha(x_k\sigma) + \vec r)$ for any
non-negative matrix $D$ and especially
$D L_1\alpha(x_1\sigma) + \cdots + D L_k\alpha(x_k\sigma) + D \vec l + A \delta 
\geqslant
D R_1\alpha(x_1\sigma) + \cdots + D R_k\alpha(x_k\sigma) + D\vec r$ 
if $A\geqslant D$ (which is no restriction since~\ALG{M} has constant
growth and any $D$ that can occur is a matrix product of the shape
$M_1\cdot \ldots \cdot M_p \leqslant A$ for some $p\in \Nat$).
The proof concludes by the observation that the above inequation implies 
$[\alpha]_\ALG{M}(C[l\sigma]) + A\delta \geqslant [\alpha]_\ALG{M}(C[r\sigma])$
for any context $C$.
\end{proof}

We conclude this section with a discussion of the above theorem.
Due to~\cite[Theorem~9]{NZM10} TMIs where each matrix
$M$ satisfies $M_{(i,i)} < 1$ for any $i \geqslant 2$ have constant growth.
Since SLIs trivially adhere to this restriction Theorem~\ref{THM:cgp_tmi} 
subsumes Corollary~\ref{COR:sli}. The next example shows that this inclusion
is strict.

\begin{exa}
Let
$\TRS{R}_1 = \{\m{a}(x) \to \m{c}(x)\}$,
$\TRS{R}_2 = \{\m{a}(\m{b}(\m{a}(x))) \to \m{a}(\m{b}(\m{b}(\m{a}(x))))\}$, and
$\TRS{S} = \varnothing$. 
Then the TMI~\ALG{M} with
\begin{xalignat*}{2}
\m{a}_\ALG{M}(\vec x) &= \BM 
1\NE 1\NE 0\NR
0\NE 0\NE 0\NR
0\NE 0\NE 0\NR
\EM + \BM
0\NR
0\NR
1\NR
\EM \vec x &
\m{b}_\ALG{M}(\vec x) &= \BM 
1\NE 0\NE 0\NR
0\NE 0\NE 1\NR
0\NE 0\NE 0\NR
\EM + \vec x 
\end{xalignat*}
where $\m{c}_\ALG{M}(\vec x) = \m{a}_\ALG{M}(\vec x)$ has constant
growth and transforms $\REL{(\TRS{R}_1\cup\TRS{R}_2)}{\TRS{S}}$ into 
$\REL{\TRS{R}_1}{(\TRS{R}_2\cup\TRS{S})}$
according to Theorem~\ref{THM:cgp_tmi}. However, there exists no SLI 
that orients the rule in~$\TRS{R}_2$ strictly which shows that
Corollary~\ref{COR:sli} cannot achieve this step.
\end{exa}

\section{Relative Match-Bounds}
\label{BOUNDS:main}

In this section we illustrate how the match-bound technique can be used to
prove relative termination and estimate complexity bounds for relative 
rewriting.
To maximize the power of the method we combine the ideas
in~\cite{W07} with the ones in~\cite{ZK10}. Preliminaries for match-bounds
are introduced in Section~\ref{BOUNDS:pre}. Section~\ref{BOUNDS:match} shows
how the technique works for linear systems before
Section~\ref{BOUNDS:raise} extends applicability to non-left-linear
systems. Automation is addressed in Section~\ref{BOUNDS:auto}.
Throughout this section we consider
$L \subseteq \FTERMS$ which does not
affect the results by assuming that the signature~\SIG{F} always contains
a constant.

\subsection{Preliminaries}
\label{BOUNDS:pre}

Let $\SIG{F}$ be a signature, $\TRS{R}$ a TRS over $\SIG{F}$, and
$L \subseteq \FTERMS$ a set of ground terms. The set
$\{t \in \FTERMS \mid \text{$s \to_\TRS{R}^* t$ for some $s \in L$}\}$
of reducts of~$L$ is denoted by $\Succ{\TRS{R}}{L}$. Given a set
$N \subseteq \NAT$ of natural numbers, the signature $\SIG{F} \times N$ is
denoted by $\SIG{F}_N$. Here function symbols $(f,c)$ with $f \in \SIG{F}$
and $c \in N$ have the same arity as~$f$ and are written as $f_c$.
The mappings $\LIFT_c\colon \SIG{F} \to \SIG{F}_\NAT$,
$\BASE\colon \SIG{F}_\NAT \to \SIG{F}$, and
$\HEIGHT\colon \SIG{F}_\NAT \to \NAT$ are defined as $\lift{f}{c} = f_c$,
$\base{f_c} = f$, and $\height{f_c} = c$ for all $f \in \SIG{F}$ and
$c \in \NAT$. They are extended to terms, sets of terms, and TRSs
in the obvious way.
The TRS $\Raise{\SIG{F}}$ over the signature $\SIG{F}_\NAT$ consists of
all rules $f_c(x_1,\ldots,x_n) \to f_{c+1}(x_1,\ldots,x_n)$ with~$f$ an
$n$-ary function symbol in $\SIG{F}$, $c \in \NAT$, and $x_1,\ldots,x_n$
pairwise distinct variables. The restriction of $\Raise{\SIG{F}}$
to the signature $\SIG{F}_{\{0,\dots,c\}}$ is denoted by $\RAISE_c(\SIG{F})$.
For terms $s,t \in \TERMS{\SIG{F}_\NAT}{\VAR{V}}$ we write
$s \uparrow t$ for the least term~$u$ with
$s \to_{\Raise{\SIG{F}}}^* u$ and $t \to_{\Raise{\SIG{F}}}^* u$.
Here least refers to the (sum of the) lengths of the joining sequences.
We extend this notion to ${\uparrow} S$ for finite non-empty sets
$S \subseteq \TERMS{\SIG{F}_\NAT}{\VAR{V}}$ in the obvious way.
Note that ${\uparrow} S$ is undefined whenever~$S$ contains
two terms~$s$ and~$t$ such that $\base{s} \neq \base{t}$.
The TRS $\match{\TRS{R}}$ over the signature $\SIG{F}_\NAT$ consists of
all rewrite rules $l' \to \lift{r}{c}$ for which there exists a rule
$l \to r \in \TRS{R}$ such that $\base{l'} = l$ and
$c = 1 + \min \{\height{l'(p)} \mid p \in \FPos(l)\}$. Here $c \in \NAT$.
The restriction of $\match{\TRS{R}}$ to the signature $\SIG{F}_{\{0,\dots,c\}}$
is denoted by $\MATCH_c(\TRS{R})$.
To be able to apply the match-bound technique to non-left-linear TRSs we
define the relation $\rto_{\match{\TRS{R}}}$ on $\TERMS{\SIG{F}_\NAT}{\VAR{V}}$
as follows: $s \rto_{\match{\TRS{R}}} t$ if and only if there exist a rewrite
rule $l \to r \in \match{\TRS{R}}$, a position $p \in \Pos(s)$, a context~$C$,
and terms $s_1,\ldots,s_n$ such that $l = C[x_1,\ldots,x_n]$ with all variables
displayed, $s|_p = C[s_1,\ldots,s_n]$, $\base{s_i} = \base{s_j}$ whenever
$x_i = x_j$ for all $i,j \in \{1,\dots,n\}$, and $t = s[r\sigma]_p$. Here
the substitution $\sigma$ is defined as follows:
\[
\sigma(x) =
\begin{cases}
{\uparrow} \{s_i \mid \text{$x_i = x$ with $i \in \{1,\dots,n\}$}\} &
\text{if $x \in \{x_1,\ldots,x_n\}$} \\
x & \text{otherwise}
\end{cases}
\]
Let~$L$ be a set of ground terms. A TRS~$\TRS{R}$ is called \emph{match-bounded}
for~$L$ if there exists a $c \in \NAT$ such that the maximum height of
function symbols occurring in terms in $\Succ{\match{\TRS{R}}}{\lift{L}{0}}$
is at most~$c$. Similarly, a TRS~$\TRS{R}$ is called
\emph{match-raise-bounded} for~$L$ if there exists a $c \in \NAT$ such
that the maximum height of function symbols occurring in terms belonging
to $\Succ[\rto]{\match{\TRS{R}}}{\lift{L}{0}}$ is at most~$c$. If we want
to make the bound~$c$ precise, we say that $\TRS{R}$ is match(-raise)-bounded
for~$L$ \emph{by~$c$}. If we do not specify the set of terms~$L$ then it
is assumed that $L = \FTERMS$. The main result underlying the match-bound
technique states that a TRS $\TRS{R}$ is terminating for a language~$L$
if~$\TRS{R}$ is linear and match-bounded for~$L$ or~$\TRS{R}$ is
non-duplicating and match-raise-bounded for~$L$.

In order to prove that a TRS $\TRS{R}$ is match(-raise)-bounded for
some language~$L$, the idea is to construct a (quasi-deterministic
and raise-consistent) tree automaton that is compatible with
$\match{\TRS{R}}$ and $\lift{L}{0}$. In the following we briefly
recall the most important definitions in this connection.
A \emph{tree automaton} $\TA{A} = (\SIG{F},Q,Q_f,\Delta)$ consists of a
signature~$\SIG{F}$, a finite set of states~$Q$, a set of final states
$Q_f \subseteq Q$, and a set of transitions $\Delta$ of the form
$f(q_1,\ldots,q_n) \to q$ or $p \to q$ where~$f$ is an $n$-ary function symbol
in $\SIG{F}$ and $p,q,q_1,\ldots,q_n \in Q$. The language $\Lg{\TA{A}}$ of
$\TA{A}$ is the set of ground terms $t \in \GTERMS{\SIG{F}}$ such
that $t \to_\Delta^* q$ for some $q \in Q_f$.
We say that $\TA{A}$ is \emph{compatible} with a TRS $\TRS{R}$ and
a language~$L$ if $L \subseteq \Lg{\TA{A}}$ and for each rewrite rule
$l \to r \in \TRS{R}$ and state substitution $\sigma\colon \Var(l) \to Q$
such that $l\sigma \to_\Delta^* q$ it holds that $r\sigma \to_\Delta^* q$.
For left-linear $\TRS{R}$ it is known that
$\Succ{\TRS{R}}{L} \subseteq \Lg{\TA{A}}$ whenever $\TA{A}$ is compatible
with $\TRS{R}$ and~$L$~\cite{G98}. To obtain a similar result for
non-left-linear TRSs, in~\cite{KM07} quasi-deterministic automata
are introduced. Let $\TA{A} = (\SIG{F},Q,Q_f,\Delta)$ be a tree automaton.
We say that a state~$p$ \emph{subsumes} a state~$q$ if~$p$ is final when
$q$ is final and for all transitions
$f(u_1,\dots,q,\dots,u_n) \to u \in \Delta$, the transition
$f(u_1,\dots,p,\dots,u_n) \to u$ belongs to $\Delta$. For a
left-hand side $l \in \lhs{\Delta}$ of a transition, the set
$\{q \mid l \to q \in \Delta\}$ of possible right-hand sides is denoted
by $Q(l)$. The automaton $\TA{A}$ is said to be \emph{quasi-deterministic}
if for every $l \in \lhs{\Delta}$ there exists a state $p \in Q(l)$ which
subsumes every other state in $Q(l)$. In general, $Q(l)$ may contain more
than one state that satisfies the above property. In the following we
assume that there is a unique designated state in $Q(l)$, which we denote by
$p_l$. The set of all designated states is denoted by $Q_d$ and the
restriction of $\Delta$ to transitions $l \to q$ that satisfy $q = p_l$
is denoted by $\Delta_d$. In~\cite{KM07} it is
shown that the tree automaton
induced by $\Delta_d$ is deterministic and accepts the same language as
$\TA{A}$. For non-left-linear TRSs $\TRS{R}$ we modify the above definition
of compatibility by demanding that the tree automaton $\TA{A}$ is
quasi-deterministic and for each rewrite rule $l \to r \in \TRS{R}$
and state substitution $\sigma\colon \Var(l) \to Q_d$ with
$l\sigma \to_{\Delta_d}^* q$ it holds that $r\sigma \to_{\Delta}^* q$.
To ensure that quasi-deterministic and compatible tree automata can be
used to prove match-raise-boundedness of a TRS $\TRS{R}$ it must be
guaranteed that the obtained tree automata are closed under the implicit
raise-steps caused by the relation $\rto_{\match{\TRS{R}}}$. To this end
we additionally require that the resulting tree automata fulfill the
property defined below. Let $\TA{A} = (\SIG{F}_N,Q,Q_f,\Delta)$ be a
tree automaton with~$N$ a finite subset of $\NAT$. We say that $\TA{A}$
is \emph{raise-consistent} if for every transition
$f_c(q_1,\ldots,q_n) \to q \in \Delta$ and left-hand side
$f_d(q_1,\ldots,q_n) \in \lhs{\Delta}$ with $c <_\Nat d$, the
transition $f_d(q_1,\ldots,q_n) \to q$ belongs to $\Delta$.

By a remark in~\cite{GHWZ07} we know that the derivation height of a
term in~$L$ is at most linear in the size of the term whenever $\TRS{R}$
is match-bounded for~$L$.
It is easy to extend this result to match-raise-boundedness and hence to
non-duplicating TRSs. To this end we need the following notions.
Let $\Mul{\NAT}$ denote the set of all finite multisets over $\NAT$. For
any $M \in \Mul{\NAT}$ we write $M(n)$ to denote how often the number
$n \in \NAT$ occurs in~$M$. Let $M,N \in \Mul{\NAT}$ be two multisets.
We write $M \cup N$ for the multiset sum of~$M$ and~$N$ where
$(M \cup N)(n) = M(n) + N(n)$ for all $n \in \NAT$ and $M \subseteq N$
for the multiset inclusion, i.e., $M(n) \leqslant N(n)$ for all
$n \in \NAT$. The multiset difference $M \setminus N$ is defined as
$(M \setminus N)(n) = M(n) - N(n)$ if $M(n) > N(n)$ and
$(M \setminus N)(n) = 0$ otherwise, for all $n \in \NAT$. We write
$M \mgt N$ if there are multisets~$X$ and $Y$ such that
$N = (M \setminus X) \cup Y$, $X \neq \varnothing$, and
for all $m \in Y$ there is an $n \in X$ such that $n <_\Nat m$.
We write $M \mge N$ if $M \mgt N$ or $M = N$. Let $\SIG{F}$ be some
signature. We extend the orderings $\mgt$ and $\mge$ to terms over
the signature $\SIG{F}_\NAT$ as follows: we have $s \mgt t$
if $\MFun{s} \mgt \MFun{t}$ and $s \mge t$ if $\MFun{s} \mge \MFun{t}$
for terms $s,t \in \TERMS{\SIG{F}_\NAT}{\VAR{V}}$. Here
$\MFun{t}  = \{\height{t(p)} \mid p \in \FPos(t)\}$ denotes
the multiset of the heights of function symbols occurring in
the term~$t$.

\begin{thm}
\label{THM:match(-raise)-bounds => linear complexity}
Let $\TRS{R}$ be a TRS and~$L$ be a language. If $\TRS{R}$ is
linear and
match-bounded or non-duplicating and match-raise-bounded for~$L$ then
$\cp{n}{\to_\TRS{R}}{L} = \OO(n)$.
\end{thm}
\begin{proof}
Assume that $\TRS{R}$ is match-raise-bounded for~$L$ and hence terminating
on~$L$. (Note that for a linear TRS $\TRS{R}$, match-boundedness coincides
with match-raise-boundedness.) Let
\[
t
\to_\TRS{R} t_1
\to_\TRS{R} \cdots
\to_\TRS{R} t_{m-1}
\to_\TRS{R} t_m
\]
be an arbitrary (terminating) rewrite sequence with $t \in L$. Since
every $\to_\TRS{R}$ rewrite sequence can be lifted to a $\rto_{\match{\TRS{R}}}$
rewrite sequence~\cite[Lemma~12]{KM09} we obtain a derivation
\[
t'
\rto_{\match{\TRS{R}}} t_1'
\rto_{\match{\TRS{R}}} \cdots
\rto_{\match{\TRS{R}}} t_{m-1}'
\rto_{\match{\TRS{R}}} t_m'
\]
such that $t' = \lift{t}{0}$ and $\base{t_i'} = t_i$ for all
$i \in \{1,\dots,m\}$. From the proof of \cite[Lemma~8]{KM09} we know
that for any non-duplicating TRS $\TRS{R}$ we have
${\rto_{\match{\TRS{R}}}} \subseteq {\mgt}$. It follows that
$t_i' \mgt t_{i+1}'$ for all $i \in \{0,\dots,m-1\}$.
Here $t_0' = t'$. Since $\TRS{R}$ is match-raise-bounded for~$L$, all
terms in this latter sequence belong to $\GTERMS{\SIG{F}_{\{0,\dots,c\}}}$
for some $c \in \NAT$. Let~$k$ be the maximal number of function symbols
occurring in some right-hand side in $\TRS{R}$. Due to a remark
in~\cite{DM79} we know that the length of the $\mgt$ chain
from $t'$ to $t_m'$ is bounded by $\|t'\| \cdot (k+1)^c$. Since
$\|t'\| = \|t\|$ and the $\mgt$ chain starting at $t'$ is at least
as long as the lifted and hence original rewrite sequence, we conclude
that the length of the $\TRS{R}$-rewrite sequence starting at the
term~$t$ is bounded by $\|t\| \cdot (k+1)^c$.
\end{proof}

Based on Theorem~\ref{THM:match(-raise)-bounds => linear complexity}
it is easy to use the match-bound technique to estimate the complexity
of a relative TRS $\REL{\TRS{R}}{\TRS{S}}$;
just check for match(-raise)-boundedness of $\TRS{R} \cup \TRS{S}$.
This process either succeeds by proving that the combined TRS is
match(-raise)-bounded, or, when $\TRS{R} \cup \TRS{S}$ cannot be proved
to be match(-raise)-bounded, it fails. Since the construction of a
(quasi-deterministic, raise-consistent, and) compatible tree
automaton does not terminate for TRSs that are not match(-raise)-bounded,
the latter situation typically does not happen. This behavior causes
a serious problem since we cannot benefit from
relative rewriting, i.e., $\REL{\TRS{R}}{\TRS{S}}$ is
match(-raise)-bounded if and only if $\TRS{R} \cup \TRS{S}$ is.
In~\cite{W07} this problem has been addressed by specifying an
upper bound on the heights that can be introduced by rewrite rules in
$\match{\TRS{S}}$. So one tries to find a $c \in \NAT$ such that the
maximum height of function symbols occurring in reductions with the TRS
$\MATCH_{c+1}(\TRS{R}) \cup \MATCH_c(\TRS{S}) \cup \lift{\TRS{S}}{c}$
is at most~$c$. If such a bound can be established we know that
$\REL{\TRS{R}}{\TRS{S}}$ is terminating and in addition that it admits
at most linear complexity. In the following we extend this approach
to better suit relative rewriting. To this end we introduce a new
enrichment $\matchRT{\TRS{R}}{\TRS{S}}{c}$ where the rewrite rules in
$\matchRT{\TRS{R}}{\TRS{S}}{c}$ which originate from size-preserving
or size-decreasing rules in $\TRS{S}$ are labeled in such a way that
they do not increase the heights of the function symbols in a contracted
redex.

To simplify the presentation we first consider linear TRSs
only. The extension to non-duplicating TRSs is explained in
Section~\ref{BOUNDS:raise}.

\subsection{RT-Bounds for Left-Linear Relative TRSs}
\label{BOUNDS:match}

As proposed in~\cite{W07} we design the new enrichment
$\matchRT{\TRS{R}}{\TRS{S}}{c}$ such that rules originating
from~\TRS{S} may introduce function symbols with height at
most~$c$. In addition (as in~\cite{ZK10}) we try to keep the
heights of the function symbols in a contracted redex if a
size-preserving or size-decreasing rewrite rule in~$\TRS{S}$
(after dropping all heights) is applied.

\begin{defi}
\label{DEF:RT enrichment}
Let $\TRS{S}$ be a TRS over a signature $\SIG{F}$ and $c \in \NAT$. The
TRS $\MATCHRT^c(\TRS{S})$ over the signature $\SIG{F}_\NAT$ consists of
all rules $l' \to \lift{r}{d}$ such that $\base{l'} \to r \in \TRS{S}$
and
\[
d =
\begin{cases}
\min\,\{c,\height{l'(\epsilon)}\}
& \text{if $\|\base{l'}\| \geqslant \|r\|$ and }\\[-0.5ex]
& \text{$\lift{\base{l'}}{\height{l'(\epsilon)}} = l'$}
\\
\min\,\{c,1 + \height{l'(p)} \mid p \in \FPos(l')\}
& \text{otherwise}
\\
\end{cases}
\]
For a relative TRS $\REL{\TRS{R}}{\TRS{S}}$ we define
$\matchRT{\TRS{R}}{\TRS{S}}{c}$ as
$\REL{\match{\TRS{R}}}{\MATCHRT^c(\TRS{S}})$.
Let~$d~\in~\NAT$. The restriction of $\MATCHRT^c(\TRS{S})$ to the signature
$\SIG{F}_{\{0,\dots,d\}}$ is denoted by $\MATCHRT_d^c(\TRS{S})$. Likewise
the relative TRS $\MATCHRT_d^c(\REL{\TRS{R}}{\TRS{S}})$
is defined as $\REL{\MATCH_d(\TRS{R})}{\MATCHRT_d^c(\TRS{S})}$.
In case $c = d$ then $\MATCHRT_d^c(\REL{\TRS{R}}{\TRS{S}})$
is abbreviated by $\MATCHRT_c(\REL{\TRS{R}}{\TRS{S}})$
and $\MATCHRT_d^c(\TRS{S}) = \MATCHRT_c(\TRS{S})$.
\end{defi}

The idea behind the requirement $\|\base{l'}\| \geqslant \|r\|$ in the above 
definition is that such rules cannot yield an increase with respect to
the multiset measure of heights.
Let us illustrate the above definition on an example.

\begin{exa}
\label{EXPL:RT enrichment}
Consider the relative TRS $\REL{\TRS{R}}{\TRS{S}}$ with
$\TRS{R}$ consisting of the rewrite rule
\begin{xalignat*}{1}
1\colon \m{rev}(x) &\to \m{rev}'(x,\m{nil})
\end{xalignat*}
and $\TRS{S}$ consisting of the rewrite rules
\begin{xalignat*}{2}
2\colon \m{rev}'(\m{nil},y) &\to y&
3\colon \m{rev}'(\m{cons}(x,y),z) &\to \m{rev}'(y,\m{cons}(x,z))
\intertext{Then the rewrite rules}
\m{rev}_0(x) &\to \m{rev}'_1(x,\m{nil}_1) &
\m{rev}_1(x) &\to \m{rev}'_2(x,\m{nil}_2) \\
\m{rev}_2(x) &\to \m{rev}'_3(x,\m{nil}_3) &
&\cdots
\intertext{belong to $\MATCH(\TRS{R})$ and
$\MATCHRT^1(\TRS{S})$ contains the rules}
\m{rev}'_0(\m{nil}_0,y) &\to y &
\m{rev}'_0(\m{cons}_0(x,y),z) &\to \m{rev}'_0(y,\m{cons}_0(x,z)) \\
\m{rev}'_0(\m{nil}_1,y) &\to y &
\m{rev}'_0(\m{cons}_1(x,y),z) &\to \m{rev}'_1(y,\m{cons}_1(x,z)) \\
&\cdots &
\m{rev}'_2(\m{cons}_1(x,y),z) &\to \m{rev}'_1(y,\m{cons}_1(x,z))
\end{xalignat*}
Both TRSs together constitute $\matchRT{\TRS{R}}{\TRS{S}}{1}$.
\end{exa}

The new enrichment $\matchRT{\TRS{R}}{\TRS{S}}{c}$ allows to prove the
complexity of the rewrite rules in~\TRS{R} relative to the rules in~\TRS{S}.

\begin{defi}
\label{DEF:match-RT-boundedness}
Let $\REL{\TRS{R}}{\TRS{S}}$ be a relative TRS. We call
$\REL{\TRS{R}}{\TRS{S}}$ \emph{match-RT-bounded} for a language~$L$ if there
exists a $c \in \NAT$ such that the height of function symbols occurring in
terms in $\Succ{\matchRT{\TRS{R}}{\TRS{S}}{c}}{\lift{L}{0}}$
is at most~$c$.
\end{defi}

An immediate consequence of the next lemma is that every
derivation in $\REL{\TRS{R}}{\TRS{S}}$ can be lifted to a
$\matchRT{\TRS{R}}{\TRS{S}}{c}$-sequence of the same length.
This result is used later on to infer termination
and complexity results for relative rewriting.

\begin{lem}
\label{LEM:simulation via match-RT enrichment}
Let $\REL{\TRS{R}}{\TRS{S}}$ be a left-linear relative TRS
and $c \in \NAT$. If $u \to_{\TRS{R}} v$ ($u \to_{\TRS{S}} v$)
then for all terms $u'$ with $\base{u'} = u$
there exists a term $v'$ such that $\base{v'} = v$ and
$u' \to_{\match{\TRS{R}}} v'$ ($u' \to_{\MATCHRT^c(\TRS{S})} v'$).
\end{lem}
\begin{proof}
Straightforward.
\end{proof}

To be able to prove that a relative TRS~\REL{\TRS{R}}{\TRS{S}}
admits a linear upper complexity bound whenever it is match-RT-bounded
for a language~$L$ we slightly modify the orderings $\mgt$ and $\mge$.
Let $M,N \in \Mul{\NAT}$ be multisets. The function $\drop{M}{n}$ removes
all occurrences of the number $n \in \NAT$ from~$M$. So for all $m \in \NAT$
we have $\drop{M}{n}(m) = 0$ if $m = n$ and $\drop{M}{n}(m) = M(m)$
otherwise. The orderings $\mgt^c$ and $\mge^c$ are defined as $M \mgt^c N$
if $\drop{M}{c} \mgt \drop{N}{c}$ and $M \mge^c N$ if
$\drop{M}{c} \mge \drop{N}{c}$. Let $\SIG{F}$ be some signature. We extend
$\mgt^c$ and $\mge^c$ to terms over the signature $\SIG{F}_\NAT$ as follows:
we have $s \mgt^c t$ if $\MFun{s} \mgt^c \MFun{t}$ and $s \mge^c t$ if
$\MFun{s} \mge^c \MFun{t}$ for terms $s,t \in \TERMS{\SIG{F}_\NAT}{\VAR{V}}$.
The basic idea behind the new orderings $\mgt^c$ and $\mge^c$ is that rewrite
rules in $\MATCHRT_c(\REL{\TRS{R}}{\TRS{S}})$ which originate from~\TRS{R} are
compatible with $\mgt^c$ and the rules originating from~\TRS{S} are compatible
with $\mge^c$. However there is one problem. If \TRS{R} contains a collapsing
rule $l \to r$ then the rule $\lift{l}{c} \to \lift{r}{c}$ appears in
$\MATCHRT_c(\REL{\TRS{R}}{\TRS{S}})$ which cannot be oriented via the ordering
$\mgt^c$ although $\lift{l}{c} \mgt \lift{r}{c}$. The problem is that
collapsing rewrite rules do not increase the heights of function symbols
in a contracted redex because the right-hand sides consist just of single
variables. To avoid this problem we assume in the following that $\TRS{R}$
is non-collapsing. For collapsing~\TRS{R} one could follow the approach
in~\cite{ZK10} which can handle collapsing rewrite rules because it does not
not use an upper bound~$c$ to limit the heights that can be introduced by the
enriched system. However, a disadvantages of this approach is that the heights
of a contracted redex are increased more often. So, apart from the collapsing
case the approach presented here is more powerful than the one introduced
in~\cite{ZK10} and completely subsumes the approach in~\cite{W07}.

\begin{lem}
\label{LEM:compatibility RT enrichment}
Let $\TRS{R}$ and $\TRS{S}$ be two non-duplicating
TRSs and $c \in \NAT$. If $\TRS{R}$ is non-collapsing then
${\to_{\MATCH_c(\TRS{R})}} \subseteq {\mgt^c}$ and
${\to_{\MATCHRT_c(\TRS{S})}} \subseteq {\mge^c}$.
\end{lem}
\begin{proof}
From the proof of~\cite[Lemma~17]{GHWZ07} we know that for a non-duplicating
TRS $\TRS{R}$ and terms~$s$ and~$t$ such that $s \to_{\MATCH_c(\TRS{R})} t$
we have $s \mgt t$. So there are multisets~$X$ and~$Y$ such that
$\MFun{t} = (\MFun{s} \setminus X) \cup Y$, $X \neq \varnothing$, and
for all $d' \in Y$ there is a $d \in X$ such that $d <_\Nat d'$.
Because $\TRS{R}$ is non-collapsing we know from the definition of
$\MATCH_c(\TRS{R})$ that there is a $d \in X$ such that
$d <_\Nat c$ and $d <_\Nat d'$ for all $d' \in Y$.
From this it follows that $\drop{\MFun{t}}{c} =
(\drop{\MFun{s}}{c} \setminus \drop{X}{c}) \cup \drop{Y}{c}$,
$\drop{X}{c} \neq \varnothing$, and for all $d' \in \drop{Y}{c}$ there is a
$d \in \drop{X}{c}$ such that $d <_\Nat d'$. As an immediate consequence we
have $\drop{\MFun{s}}{c} \mgt \drop{\MFun{t}}{c}$ and hence $s \mgt^c t$.

Now let~$s$ and~$t$ be terms and $l \to r$ be a rewrite rule in
$\MATCHRT_c(\TRS{S})$ such that $s \to_{\{l \to r\}} t$. According to
Definition~\ref{DEF:RT enrichment} we have to consider two cases. The first 
case amounts to $\|l\| \geqslant \|r\|$ where all function symbols in $l$
and $r$ have the same heights. But then non-duplication of~\TRS{S}
implies $\MFun{s} \supseteq \MFun{t}$ and thus $s \mge^c t$. In the other
case if $l\to r$ is non-collapsing and $l \notin \lift{\base{l}}{c}$ then we
obtain $s \mgt^c t$ as before and hence also $s \mge^c t$. 
If $l \in \lift{\base{l}}{c}$ then 
$\drop{\MFun{s}}{c} \supseteq \drop{\MFun{t}}{c}$ since
$\drop{\MFun{l}}{c} = \drop{\MFun{r}}{c} = \varnothing$
and if $l \to r$ is collapsing then $\MFun{s} \supseteq \MFun{t}$ 
since $\MFun{r} = \varnothing$.
Hence in both situations $s \mge^c t$.
\end{proof}

Since the length of every $\mgt^c$ chain is bounded by a function linear in
the size of the starting term---if the size-increase of the terms in the
chain can be bounded by a constant---we can prove that the complexity induced
by the relative TRS $\REL{\TRS{R}}{\TRS{S}}$ on some language~$L$ is at most
linear if $\REL{\TRS{R}}{\TRS{S}}$ is match-RT-bounded for~$L$.

\begin{thm}
\label{THM:match-RT-bounds => linear complexity}
Let $\REL{\TRS{R}}{\TRS{S}}$ be a linear relative TRS and~\TRS{R} be
non-collapsing. If $\REL{\TRS{R}}{\TRS{S}}$ is match-RT-bounded for
a language~$L$ then $\REL{\TRS{R}}{\TRS{S}}$ is terminating on~$L$
and $\cp{n}{\to_{\REL{\TRS{R}}{\TRS{S}}}}{L} = \OO(n)$.
\end{thm}
\begin{proof}
First we show that $\REL{\TRS{R}}{\TRS{S}}$ is terminating on~$L$. Assume
to the contrary that there is an infinite rewrite sequence of the form
\[
t_1
\to_{\REL{\TRS{R}}{\TRS{S}}} t_2
\to_{\REL{\TRS{R}}{\TRS{S}}} t_3
\to_{\REL{\TRS{R}}{\TRS{S}}} \cdots
\]
with $t_1 \in L$. Because $\TRS{R} \cup \TRS{S}$ is left-linear and
$\REL{\TRS{R}}{\TRS{S}}$ is match-RT-bounded for~$L$ by a $c \in \NAT$,
according to Lemma~\ref{LEM:simulation via match-RT enrichment},
the above derivation can be lifted to an infinite
$\matchRT{\TRS{R}}{\TRS{S}}{c}$
rewrite sequence
\[
t_1'
\to_{\matchRT{\TRS{R}}{\TRS{S}}{c}} t_2'
\to_{\matchRT{\TRS{R}}{\TRS{S}}{c}} t_3'
\to_{\matchRT{\TRS{R}}{\TRS{S}}{c}} \cdots
\]
starting from $t_1' = \lift{t_1}{0}$ such that $\base{t_i'} = t_i$
for all $i \geqslant 1$ and the height of every function symbol occurring
in a term in the lifted sequence is at most~$c$. Hence the employed
rewrite rules in the derivation emanating from $t_1'$ must come from
$\MATCHRT_c(\REL{\TRS{R}}{\TRS{S}})$. With help of
Lemma~\ref{LEM:compatibility RT enrichment}, transitivity of $\mge^c$,
and compatibility of the orderings $\mgt^c$ and $\mge^c$ we deduce that
$t_i' \mgt^c t_{i+1}'$ for all $i \geqslant 1$. However, this is excluded
because~$<_\Nat$ is well-founded on $\{0,\dots,c\}$ and hence $\mgt^c$ is
well-founded on $\TERMS{\SIG{F}_{\{0,\dots,c\}}}{\VAR{V}}$.

To prove the second part of the theorem, consider an arbitrary
(terminating) rewrite sequence
\[
u
\to_{\REL{\TRS{R}}{\TRS{S}}} u_1
\to_{\REL{\TRS{R}}{\TRS{S}}} \cdots
\to_{\REL{\TRS{R}}{\TRS{S}}} u_m
\]
with $u \in L$. Similar as before this rewrite sequence can be lifted
to a $\matchRT{\TRS{R}}{\TRS{S}}{c}$-sequence of the same length
\[
u'
\to_{\matchRT{\TRS{R}}{\TRS{S}}{c}} u_1'
\to_{\matchRT{\TRS{R}}{\TRS{S}}{c}} \cdots
\to_{\matchRT{\TRS{R}}{\TRS{S}}{c}} u_m'
\]
such that $u' = \lift{u}{0}$ and $u_i' \mgt^c u_{i+1}'$
for all $i \in \{0,\dots,m-1\}$. Here $u_0' = u'$ and $c \in \NAT$ such
that the relative TRS $\REL{\TRS{R}}{\TRS{S}}$ is match-RT-bounded for~$L$
by~$c$. Similar as in the proof of
Theorem~\ref{THM:match(-raise)-bounds => linear complexity} we
can conclude that the length of the $\REL{\TRS{R}}{\TRS{S}}$-rewrite
sequence starting at the term~$u$ is bounded by
$\|u\| \cdot (k+1)^c$ where~$k$ is the maximal number of function
symbols occurring in some right-hand side in $\TRS{R} \cup \TRS{S}$;
just replace $\mgt$ by $\mgt^c$.
\end{proof}

We conclude this subsection with an example.

\begin{exa}
The relative TRS $\REL{\TRS{R}}{\TRS{S}}$
of Example~\ref{EXPL:RT enrichment}
is match-RT-bounded for $\GTERMS{\SIG{F}}$ by~$1$.
Here $\SIG{F} = \{\m{nil},\m{cons},\m{rev},\m{rev}'\}$. Due to
Theorem~\ref{THM:match-RT-bounds => linear complexity} we can
conclude that $\REL{\TRS{R}}{\TRS{S}}$ admits at most linear
derivational complexity. In Section~\ref{BOUNDS:auto} it is
explained how match-RT-boundedness can be checked automatically.
\end{exa}

\subsection{Raise-RT-Bounds for Non-Left-Linear Relative TRSs}
\label{BOUNDS:raise}

In order to generalize Theorem~\ref{THM:match-RT-bounds => linear complexity}
to non-duplicating relative TRSs we consider the relation
$\rto_{\matchRT{\TRS{R}}{\TRS{S}}{c}}$ instead of
$\to_{\matchRT{\TRS{R}}{\TRS{S}}{c}}$ which uses raise-rules
to deal with non-left-linearity. Thereby the rewrite relation
$\rto_{\matchRT{\TRS{R}}{\TRS{S}}{c}}$ is defined as
$\rto_{\MATCHRT^c(\TRS{S})}^*
\cdot \rto_{\match{\TRS{R}}}
\cdot \rto_{\MATCHRT^c(\TRS{S})}^*$
where $\rto_{\MATCHRT^c(\TRS{S})}$ is
defined similar to $\rto_{\match{\TRS{R}}}$ (but based on
$\MATCHRT^c(\TRS{S})$ instead of $\match{\TRS{R}}$).
This is essential to lift rewrite sequences in the relative TRS
$\REL{\TRS{R}}{\TRS{S}}$ to sequences in $\matchRT{\TRS{R}}{\TRS{S}}{c}$.

\begin{defi}
Let $\REL{\TRS{R}}{\TRS{S}}$ be a relative TRS.
We call $\REL{\TRS{R}}{\TRS{S}}$ \emph{match-raise-RT-bounded}
for a language~$L$ if there exists a number $c \in \NAT$ such that
the height of function symbols occurring in terms belonging to
$\Succ[\rto]{\matchRT{\TRS{R}}{\TRS{S}}{c}}{\lift{L}{0}}$
is at most~$c$.
\end{defi}

Note that for left-linear relative TRSs, match-raise-RT-boundedness
coincides with match-RT-boundedness. By using the relation
$\rto_{\matchRT{\TRS{R}}{\TRS{S}}{c}}$ every derivation
induced by the relative TRS $\REL{\TRS{R}}{\TRS{S}}$ can be
simulated via the rewrite rules in $\matchRT{\TRS{R}}{\TRS{S}}{c}$.

\begin{lem}
\label{LEM:simulation via match-raise-RT enrichment}
Let $\REL{\TRS{R}}{\TRS{S}}$ be a relative TRS and $c \in \NAT$. 
If $u \to_{\TRS{R}} v$ ($u \to_{\TRS{S}} v)$ then for all terms
$u'$ with $\base{u'} = u$
there exists a term $v'$ such that $\base{v'} = v$ and
$u' \rto_{\match{\TRS{R}}} v'$ ($u' \rto_{\MATCHRT^c(\TRS{S})} v'$).
\end{lem}
\begin{proof}
Straightforward.
\end{proof}

Before we can prove that match-raise-RT-boundedness of \REL{\TRS{R}}{\TRS{S}}
induces a linear upper bound on the complexity we have to ensure that the
raise-rules implicitly used by the relation
$\rto_{\matchRT{\TRS{R}}{\TRS{S}}{c}}$
can be oriented via $\mge^c$.

\begin{lem}
\label{LEM:compatibility raise}
For any signature $\SIG{F}$ and $c \in \NAT$ it holds that
${\to_{\RAISE_c(\SIG{F})}} \subseteq {\mge^c}$.
\end{lem}
\begin{proof}
Assume that there are terms~$s$ and~$t$ such that
$s \to_{\RAISE_c(\SIG{F})} t$. According to the definition of
$\RAISE_c(\SIG{F})$ we have
$\MFun{t} = (\MFun{s} \setminus \{d\}) \cup \{d+1\}$ for some
height $d <_\Nat c$. Thus $s \mgt^c t$ and hence $s \mge^c t$
according to the definition of $\mge^c$.
\end{proof}

Using Lemma~\ref{LEM:compatibility raise} it is easy to extend
Theorem~\ref{THM:match-RT-bounds => linear complexity} to
TRSs that are non-linear but non-duplicating.

\begin{thm}
\label{THM:match-raise-RT-bounds => linear complexity}
Let $\REL{\TRS{R}}{\TRS{S}}$ be a non-duplicating relative
TRS and \TRS{R} be non-collapsing. If $\REL{\TRS{R}}{\TRS{S}}$
is match-raise-RT-bounded for a language~$L$ then
$\REL{\TRS{R}}{\TRS{S}}$ is terminating on~$L$. Furthermore,
$\cp{n}{\to_{\REL{\TRS{R}}{\TRS{S}}}}{L} = \OO(n)$.
\end{thm}
\begin{proof}
First we show that $\REL{\TRS{R}}{\TRS{S}}$ is terminating on~$L$.
Assume to the contrary that there is an infinite rewrite sequence
of the form
\[
t_1
\to_{\REL{\TRS{R}}{\TRS{S}}} t_2
\to_{\REL{\TRS{R}}{\TRS{S}}} t_3
\to_{\REL{\TRS{R}}{\TRS{S}}} \cdots
\]
with $t_1 \in L$. Let $\REL{\TRS{R}}{\TRS{S}}$ be
match-raise-RT-bounded for~$L$ by a $c \in \NAT$.
Lemma~\ref{LEM:simulation via match-raise-RT enrichment} yields
an infinite $\rto_{\matchRT{\TRS{R}}{\TRS{S}}{c}}$ rewrite sequence
\[
t_1'
\rto_{\matchRT{\TRS{R}}{\TRS{S}}{c}} t_2'
\rto_{\matchRT{\TRS{R}}{\TRS{S}}{c}} t_3'
\rto_{\matchRT{\TRS{R}}{\TRS{S}}{c}} \cdots
\]
starting from $t_1' = \lift{t_1}{0}$ such that $\base{t_i'} = t_i$
for all $i \geqslant 1$. Because $\REL{\TRS{R}}{\TRS{S}}$ is
match-raise-RT-bounded for~$L$ by~$c$, the height of every function symbol
occurring in a term in the lifted sequence is at most~$c$. Hence the
employed rewrite rules in the derivation emanating from $t_1'$ must come
from $\MATCHRT_c(\REL{\TRS{R}}{\TRS{S}})$. With help of
Lemmata~\ref{LEM:compatibility RT enrichment}
and~\ref{LEM:compatibility raise}, transitivity of $\mge^c$,
and compatibility of $\mgt^c$ and $\mge^c$ we deduce that
$t_i' \mgt^c t_{i+1}'$ for all $i \geqslant 1$. (Note that
Lemma~\ref{LEM:compatibility RT enrichment} requires that
$\REL{\TRS{R}}{\TRS{S}}$ is non-duplicating.) However, this
is excluded because~$<_\Nat$ is well-founded on $\{0,\dots,c\}$
and hence $\mgt^c$ is well-founded on
$\TERMS{\SIG{F}_{\{0,\dots,c\}}}{\VAR{V}}$.

To prove the second part of the theorem, consider an arbitrary
(terminating) rewrite sequence
\[
u
\to_{\REL{\TRS{R}}{\TRS{S}}} u_1
\to_{\REL{\TRS{R}}{\TRS{S}}} \cdots
\to_{\REL{\TRS{R}}{\TRS{S}}} u_m
\]
with $u \in L$. Similar as before this rewrite sequence can be lifted
to a $\rto_{\matchRT{\TRS{R}}{\TRS{S}}{c}}$-sequence of the same length
\[
u'
\rto_{\matchRT{\TRS{R}}{\TRS{S}}{c}} u_1'
\rto_{\matchRT{\TRS{R}}{\TRS{S}}{c}} \cdots
\rto_{\matchRT{\TRS{R}}{\TRS{S}}{c}} u_m'
\]
such that $u' = \lift{u}{0}$ and $u_i' \mgt^c u_{i+1}'$ for all
$i \in \{0,\dots,m-1\}$. Here $u_0' = u'$ and $c \in \NAT$ such that
the relative TRS $\REL{\TRS{R}}{\TRS{S}}$ is match-raise-RT-bounded
for~$L$ by~$c$. Similar as in the proof of
Theorem~\ref{THM:match(-raise)-bounds => linear complexity} we can
conclude that the length of the $\REL{\TRS{R}}{\TRS{S}}$-rewrite
sequence starting at the term~$u$ is bounded by
$\|u\| \cdot (k+1)^c$ where~$k$ is the maximal number of function
symbols occurring in some right-hand side in $\TRS{R} \cup \TRS{S}$;
just replace $\mgt$ by $\mgt^c$.
\end{proof}

\subsection{Automation}
\label{BOUNDS:auto}

To automatically prove that a given relative TRS is match(-raise)-RT-bounded
for some language~$L$ we use (quasi-deterministic, raise-consistent, and)
compatible tree automata. Here a tree automaton $\TA{A}$ is said to be
\emph{compatible} with a relative TRS
$\REL{\TRS{R}}{\TRS{S}}$ and a language~$L$ if $\TA{A}$ is compatible
with $\TRS{R} \cup \TRS{S}$ and~$L$.

\begin{lem}
Let $\REL{\TRS{R}}{\TRS{S}}$ be a left-linear relative TRS, $L$ a language,
and $c \in \NAT$. Let $\TA{A}$ be a tree automaton. If $\TA{A}$ is compatible
with the relative TRS $\matchRT{\TRS{R}}{\TRS{S}}{c}$ and $\lift{L}{0}$
such that the height of each function symbol occurring in transitions in
$\TA{A}$ is at most~$c$ then $\REL{\TRS{R}}{\TRS{S}}$ is match-RT-bounded
for~$L$.
\end{lem}
\begin{proof}
Easy consequence of Definition~\ref{DEF:match-RT-boundedness}
and the fact that compatible tree automata are closed under
left-linear rewriting.
\end{proof}

In case of non-left-linear TRSs we obtain the following result.

\begin{lem}
Let $\REL{\TRS{R}}{\TRS{S}}$ be a relative TRS, $L$ a language, and
$c \in \NAT$. Let $\TA{A}$ be a quasi-deterministic and raise-consistent
tree automaton. If $\TA{A}$ is compatible with $\matchRT{\TRS{R}}{\TRS{S}}{c}$
and $\lift{L}{0}$ such that the height of each function symbol occurring in
transitions in $\TA{A}$ is at most~$c$ then $\REL{\TRS{R}}{\TRS{S}}$ is
match-raise-RT-bounded for~$L$.
\end{lem}
\begin{proof}
Straightforward by using the fact that quasi-deterministic, raise-consistent
and compatible tree automata are closed under rewriting.
\end{proof}

To prove that a relative TRS $\REL{\TRS{R}}{\TRS{S}}$ is
match(-raise)-RT-bounded for a set of terms~$L$ we construct a
(quasi-deterministic and raise-consistent) tree automaton
$\TA{A} = (\SIG{F},Q,Q_f,\Delta)$ that is compatible with the rewrite rules
of $\matchRT{\TRS{R}}{\TRS{S}}{c}$ and $\lift{L}{0}$. Since the set
$\Succ{\matchRT{\TRS{R}}{\TRS{S}}{c}}{\lift{L}{0}}$ need not be
regular, even for left-linear $\TRS{R}$ and $\TRS{S}$ and regular~$L$
(see~\cite{GHWZ07}) we cannot hope to give an exact automaton construction.
The general idea~\cite{G98,GHWZ07} is to look for violations of the
compatibility requirement: $l\sigma \to_\Delta^* q$
($l\sigma \to_{\Delta_d}^* q$) and $r\sigma \not\to_\Delta^* q$
for some rewrite rule $l \to r$, state substitution
$\sigma\colon \Var(l) \to Q$ ($\sigma\colon \Var(l) \to Q_d$),
and state $q \in Q$ ($q \in Q_d$). Then we add new states and transitions
to the current automaton to ensure $r\sigma \to_\Delta^* q$. After
$r\sigma \to_\Delta^* q$ has been established, we repeat this process until
a (quasi-deterministic, raise-consistent, and) compatible automaton is
obtained. Note that this may never happen if new states are repeatedly added.
To guess an appropriate~$c$ we start with $c = 0$. As soon as a new transition
$f_d(q_1,\ldots,q_n) \to q$ with $d >_\Nat c$ is added to the constructed tree
automaton, we set $c = d$ and proceed with the construction.

\begin{exa}
We show that the relative TRS $\REL{\TRS{R}}{\TRS{S}}$
of Example~\ref{EXPL:RT enrichment} over the signature
$\SIG{F} = \{\m{nil},\m{cons},\m{rev},\m{rev}'\}$
is match-RT-bounded for~$\FTERMS$ by constructing a compatible
tree automaton. As starting point we consider the initial tree automaton
\begin{xalignat*}{4}
\m{nil}_0 &\to 1 &
\m{cons}_0(1,1) &\to 1 &
\m{rev}_0(1) &\to 1 &
\m{rev}'_0(1,1) &\to 1
\end{xalignat*}
which accepts all ground terms over the enriched signature $\lift{\SIG{F}}{0}$.
The first compatibility violation we consider is caused by the rewrite rule
$\m{rev}_0(x) \to_{\match{\TRS{R}}} \m{rev}'_1(x,\m{nil}_1)$. We have
$\m{rev}_0(1) \to 1$ but not $\m{rev}'_1(1,\m{nil}_1) \to^* 1$. To solve
this violation we add the transitions $\m{nil}_1 \to 2$ and
$\m{rev}'_1(1,2) \to 1$. The compatibility violation caused by the rewrite
rule $\m{rev}'_1(\m{nil}_0,y) \to_{\MATCHRT^1(\TRS{S})} y$ and
the derivation $\m{rev}'_1(\m{nil}_0,2) \to^* 1$ is solved by adding the
transition $2 \to 1$. Note that we are currently using $\MATCHRT^1(\TRS{S})$
because the maximal height of a function symbol occurring in the underlying
tree automaton is~$1$. Next we consider the compatibility violation
$\m{rev}'_1(\m{cons}_0(1,1),2) \to^* 1$ but
$\m{rev}'_1(1,\m{cons}_1(1,2)) \not\to^* 1$ induced by the rule
$\m{rev}'_1(\m{cons}_0(x,y),z) \to_{\MATCHRT^1(\TRS{S})}
\m{rev}'_1(y,\m{cons}_1(x,z))$. In order to ensure
$\m{rev}'_1(1,\m{cons}_1(1,2)) \to^* 1$ we reuse the transition
$\m{rev}'_1(1,2) \to 1$ and add the new transition
$\m{cons}_1(1,2) \to 2$. Finally,
$\m{rev}'_0(\m{cons}_1(1,2),1) \to^* 1$ and
$\m{rev}'_0(\m{cons}_1(x,y),z) \to_{\MATCHRT^1(\TRS{S})}
\m{rev}'_1(y,\m{cons}_1(x,z))$
give rise to the transition $\m{cons}_1(1,1) \to 2$. After
this step, the obtained tree automaton is compatible with
$\matchRT{\TRS{R}}{\TRS{S}}{1}$. Hence $\REL{\TRS{R}}{\TRS{S}}$
is match-RT-bounded for~$\FTERMS$ by~$1$. Due to
Theorem~\ref{THM:match-RT-bounds => linear complexity} we
can conclude that $\REL{\TRS{R}}{\TRS{S}}$ admits at most
linear complexity. We remark that the ordinary match-bound technique
(Theorem~\ref{THM:match(-raise)-bounds => linear complexity})
fails on $\REL{\TRS{R}}{\TRS{S}}$ because $\TRS{R} \cup \TRS{S}$ induces
quadratic complexity:
\begin{xalignat*}{1}
\m{rev}^n(x)\sigma^m &\to
\m{rev}^{n-1}(\m{rev}'(x,\m{nil}))\sigma^m \to^m
\m{rev}^{n-1}(\m{rev}'(\m{nil},x))\sigma^m \\
&\to \m{rev}^{n-1}(x)\sigma^m \to^{(n-1)(m+2)}
x\sigma^m
\end{xalignat*}
with $\sigma = \{x \mapsto \m{cons}(y,x)\}$) for all $n,m \geqslant 1$.
\end{exa}

\section{Assessment}
\label{ASS:main}

In this section we compare the complexity proving power
of the direct and the modular setting on a theoretical level.
Gains in power in practice are reported in Section~\ref{EXP:main}.
In the first part of this section we show that for TMIs of dimension one,
i.e.\ SLIs, in theory both approaches are equivalent but
in the general case the modular setting allows TMIs of smaller
dimensions to succeed. Since the dimension of the TMI corresponds to the
degree of the polynomial bound the modular setting allows to establish tighter
bounds. The second part of this section shows that the
modular setting is strictly more powerful than the direct one, i.e.,
there are systems where the modular setting admits a complexity proof but
all involved methods cannot succeed on its own in the direct setting.
To make the presentation easier we assume the original problems to be 
standard (in contrast to relative) TRSs. This has no effect on the results.
The next lemma states that for SLIs in theory there is no difference in power
between the two settings.

\begin{lem}
\label{LEM:pow1}
Let~$\TRS{R} = \TRS{R}_1 \cup \TRS{R}_2$ be a TRS. There exists
an~SLI~\ALG{M} compatible with~$\TRS{R}$ if and only if there
exist SLIs~$\ALG{M}_1$ and~$\ALG{M}_2$ such that
$\ALG{M}_1$ is compatible with $\REL{\TRS{R}_1}{\TRS{R}_2}$ and
$\ALG{M}_2$ is compatible with $\REL{\TRS{R}_2}{\TRS{R}_1}$.
\end{lem}
\begin{proof}
The implication from left to right obviously holds since \ALG{M} is a
suitable candidate for $\ALG{M}_1$ and $\ALG{M}_2$. For the reverse
direction we construct an SLI~\ALG{M} compatible with~\TRS{R} based
on the SLIs $\ALG{M}_1$ and $\ALG{M}_2$. 
Let $f_{\ALG{M}_1}(x_1, \dots,x_m) = x_1 + \dots + x_m + f_1$ and
$f_{\ALG{M}_2}(x_1, \dots,x_m) = x_1 + \dots + x_m + f_2$.
It is straightforward to check that 
$
f_\ALG{M}(x_1, \dots, x_m) =
x_1 + \dots+ x_m + f_1 + f_2
$
for any $f \in \SIG{F}$ yields an SLI~$\ALG{M}$ compatible with~\TRS{R}.
\end{proof}

Due to Theorem~\ref{THM:modeq} the complexity is not affected when using
the modular setting. Hence when using SLIs in theory both approaches can
prove the same bounds. But experiments in Section~\ref{EXP:main} show that
in practice proofs are easier to find in the modular setting since, e.g.,
the coefficients of the interpretations can be chosen smaller (cf.\ the
proof of Lemma~\ref{LEM:pow1}).
If TMIs of larger dimensions are applied then just the only-if direction
of Lemma~\ref{LEM:pow1} holds. This is shown with the help of the next
example.

\begin{exa}
\label{EX:power_tmi}
Consider the TRS~\TRS{R} (\tpdb{Strategy\_removed\_AG01/\#4.21})
consisting of the rules:
\begin{xalignat*}{4}
1\colon\m{f}(\m{1})    &\to \m{f}(\m{g}(\m{1}))&
2\colon\m{f}(\m{f}(x)) &\to \m{f}(x)&
3\colon\m{g}(\m{0})    &\to \m{g}(\m{f}(\m{0}))&
4\colon\m{g}(\m{g}(x)) &\to \m{g}(x)
\end{xalignat*}
The TMIs~$\ALG{M}_1$
\begin{xalignat*}{4}
\m{f}_{\ALG{M}_1}(\vec x) &= \BM
1\NE 1\NR
0\NE 0\NR
\EM \vec x +
\BM
0\NR
1\NR
\EM
&
\m{g}_{\ALG{M}_1}(\vec x) &= \BM
1\NE 0\NR
0\NE 0\NR
\EM \vec x
&
\m{0}_{\ALG{M}_1} &= \BM
0\NR
0\NR
\EM
&
\m{1}_{\ALG{M}_1} &= \BM
0\NR
1\NR
\EM
\intertext{and~$\ALG{M}_2$}
\m{g}_{\ALG{M}_2}(\vec x) &= \BM
1\NE 1\NR
0\NE 0\NR
\EM \vec x +
\BM
0\NR
1\NR
\EM
&
\m{f}_{\ALG{M}_2}(\vec x) &= \BM
1\NE 0\NR
0\NE 0\NR
\EM \vec x
&
\m{1}_{\ALG{M}_2} &= \BM
0\NR
0\NR
\EM
&
\m{0}_{\ALG{M}_2} &= \BM
0\NR
1\NR
\EM
\end{xalignat*}
show quadratic upper bounds on the derivational complexity of the systems
$\REL{\{3,4\}}{\{1,2\}}$ and $\REL{\{1,2\}}{\{3,4\}}$,
respectively. Theorem~\ref{THM:modeq} establishes a quadratic upper bound
for \TRS{R}.
\end{exa}

Although for the TRS in Example~\ref{EX:power_tmi} TMIs of dimension two
could establish a quadratic upper bound on the derivational complexity
in the modular setting, they cannot do so in the direct setting because
of the next lemma. (We remark that there exist TMIs of
dimension three that are compatible with this TRS).

\begin{lem}
\label{LEM:4.21}
The TRS~\tpdb{Strategy\_removed\_AG01/\#4.21} does not admit a TMI of
dimension two compatible with it.
\end{lem}
\begin{proof}
To address all possible interpretations we extracted the
set of constraints that represent a TMI of
dimension two compatible with the TRS~\tpdb{Strategy\_removed\_AG01/\#4.21}.
\MINISMT~\cite{ZM10} can detect unsatisfiability of these constraints.
Details of this proof can be found at the web site in Footnote~\ref{FOO:web}
on page~\pageref{FOO:web}.
\end{proof}

The next result shows that any direct proof transfers into the modular
setting without increasing the bounds on the complexity.

\begin{lem}
\label{LEM:pow}
Let~$\relgt$ be a finitely branching rewrite relation
and let~$\TRS{R} = \TRS{R}_1 \cup \TRS{R}_2$ be a TRS
compatible with~$\relgt$.
Then there exist complexity pairs $(\relgt_1,\relge_1)$ and
$(\relgt_2,\relge_2)$ which are
compatible with the relative TRSs
$\REL{\TRS{R}_1}{\TRS{R}_2}$ and
$\REL{\TRS{R}_2}{\TRS{R}_1}$, respectively.
Furthermore for any language~$L$ we have
$\cp{n}{\relgt}{L} = \Theta(\cp{n}{\relgt_1}{L} + \cp{n}{\relgt_2}{L})$.
\end{lem}
\begin{proof}
Fix~$i$. Let $(\relgt_i,\relge_i)$ be $(\relgt,=)$.
It is easy to see that $(\relgt,=)$ is a complexity pair because $\relgt$ and
$=$ are compatible rewrite relations. It remains to show that
for any term $t\in L$ we have
$\cp{n}{\relgt}{L} = \Theta(\cp{n}{\relgt_1}{L} + \cp{n}{\relgt_2}{L})$.
To this end we observe that
$\dl{t}{\relgt_1} + \dl{t}{\relgt_2} =
2\cdot\dl{t}{\relgt}$ for all terms~$t\in L$.
Basic properties of the $\OO$-notation yield the desired result.
\end{proof}

Due to Example~\ref{EX:power_tmi} and Lemmata~\ref{LEM:4.21} and~\ref{LEM:pow}
we obtain that the modular setting allows to use TMIs of smaller dimensions
than the direct one, which allows to establish tighter bounds. The next
example (together with Lemma~\ref{LEM:pow}) shows that in theory the modular
complexity setting is strictly more powerful than the direct one since it
allows to combine different criteria to establish an upper complexity bound
while any method on its own cannot succeed.

\begin{exa}
\label{EX:trafo}
Consider the TRS~\TRS{R} (\tpdb{Transformed\_CSR\_04/Ex16\_Luc06\_GM})
consisting of the rules:
\begin{xalignat*}{4}
1\colon\m{c} &\to \m{a} &
3\colon\m{mark}(\m{a}) &\to \m{a} &
5\colon\m{g}(x,y) &\to \m{f}(x,y)
\\
2\colon\m{c} &\to \m{b} &
4\colon\m{mark}(\m{b}) &\to \m{c} &
6\colon\m{g}(x,x) &\to \m{g}(\m{a},\m{b}) &
7\colon\m{mark}(\m{f}(x,y)) &\to \m{g}(\m{mark}(x),y)
\end{xalignat*}
The following SLI~$\ALG{M}$
\begin{xalignat*}{6}
\m{a}_\ALG{M} &= 0&
\m{b}_\ALG{M} &= 0&
\m{c}_\ALG{M} &= 1&
\m{f}_\ALG{M}(x,y) &= x+y&
\m{g}_\ALG{M}(x,y) &= x+y+1&
\m{mark}_\ALG{M}(x) &= x+2
\end{xalignat*}
allows Corollary~\ref{COR:sli} to transform the TRS~\TRS{R}
into the relative TRS $\REL{\{6,7\}}{\{1,2,3,4,5\}}$.
This problem can be split according to
Theorem~\ref{THM:modeq} into the two relative TRSs
$\REL{\{6\}}{\{1,2,3,4,5,7\}}$ and
$\REL{\{7\}}{\{1,2,3,4,5,6\}}$.
Match-bounds (Theorem~\ref{THM:match-raise-RT-bounds => linear complexity})
can show a linear upper bound on the first problem. The following
TMI~$\ALG{M}$
\begin{xalignat*}{3}
\m{a}_\ALG{M}  &= \BM
0\NR
0\NR
\EM
&
\m{f}_\ALG{M}(\vec x,\vec y) &=
\vec x +
\BM
1\NE 0\NR
0\NE 0\NR
\EM
\vec y +
\BM
0\NR
1\NR
\EM
&
\m{mark}_\ALG{M}(\vec x) &= \BM
1 \NE 1\NR
0 \NE 1\NR
\EM
\vec x
\end{xalignat*}
where $\m{a}_\ALG{M} = \m{b}_\ALG{M} = \m{c}_\ALG{M}$ and
$\m{f}_\ALG{M}(\vec x,\vec y) = \m{g}_\ALG{M}(\vec x,\vec y)$
gives a quadratic upper bound on the second relative TRS,
establishing a quadratic upper bound on the derivational complexity
of \TRS{R}. The quadratic bound is tight as $\TRS{R}$ admits derivations
\[
\m{mark}^n(x)\sigma^m \to^m
\m{mark}^{n-1}(x)\tau^m\gamma \to^m
\m{mark}^{n-1}(x)\sigma^m\gamma \to^{2m(n-1)}
x\sigma^m\gamma^n
\]
of length $2mn$ where
$\sigma = \{x \mapsto \m{f}(x,y)\}$,
$\tau = \{x \mapsto \m{g}(x,y)\}$, and
$\gamma = \{x \mapsto \m{mark}(x)\}$.
Last but not least we remark that none of the involved techniques
can establish an upper bound on its own. In case of match-bounds this
follows from the fact that~\TRS{R} admits quadratic derivational
complexity. The same reason also holds for Corollary~\ref{COR:sli}
because SLIs induce linear complexity bounds. Finally,
TMIs fail because they cannot orient the rewrite rule
$\m{g}(x,x) \to \m{g}(\m{a},\m{b})$.
\end{exa}

Hence we obtain the following corollary.

\begin{cor}
\label{COR:subsume}
The modular complexity setting is strictly more powerful than the direct~one.
\end{cor}
\begin{proof}
By Lemma~\ref{LEM:pow} and Example~\ref{EX:trafo}.
\end{proof}

Next we consider the TRS \tpdb{Zantema\_04/z086}. The question about the
derivational complexity of it has been stated as problem \#105 on the RTA
LooP.\footnote{\ \url{http://rtaloop.mancoosi.univ-paris-diderot.fr}}

\begin{exa}
Consider the TRS~\TRS{R} (\tpdb{Zantema\_04/z086}) consisting of the
rules:
\begin{xalignat*}{3}
1\colon\m{a}(\m{a}(x)) &\to \m{c}(\m{b}(x))&
2\colon\m{b}(\m{b}(x)) &\to \m{c}(\m{a}(x))&
3\colon\m{c}(\m{c}(x)) &\to \m{b}(\m{a}(x))
\end{xalignat*}
Adian~\cite{A09} showed that $\TRS{R}$ admits at most quadratic
derivational complexity. Since the proof is based on a low-level
reasoning on the structure of~\TRS{R}, it is specific to this TRS
and challenging for automation. With our approach we cannot prove
the quadratic bound on the derivational complexity of $\TRS{R}$.
However, Corollary~\ref{COR:sli} permits to establish some progress.
Using an SLI counting $\m{a}$'s and $\m{b}$'s, it suffices to
determine the derivational complexity of $\REL{\{3\}}{\{1,2\}}$.
This means that $\m{c}(\m{c}(x)) \to \m{b}(\m{a}(x))$ relative to
the other rules dominates the derivational complexity of~\TRS{R}.
The benefit is that now, e.g., a TMI must only orient one rule
strictly and the other two rules weakly (compared to all three
rules strictly).
It has to be clarified if
the relative formulation of the problem can be used to simplify the proof
in~\cite{A09}.
\end{exa}

The next example shows that although the modular approach often allows to
establish lower bounds compared to the direct one, further criteria for
splitting TRSs should be investigated.

\begin{exa}
\label{EX:4.30}
Consider the TRS~\TRS{R} (\tpdb{SK90/4.30})
consisting of the following rules:
\begin{xalignat*}{3}
1\colon\,\m{f}(\m{nil}) &\to \m{nil} &
3\colon\,\m{f}(\m{nil} \app y) &\to \m{nil} \app \m{f}(y) &
5\colon\,\m{f}((x \app y) \app z) &\to \m{f}(x \app (y \app z))\\
2\colon\m{g}(\m{nil}) &\to \m{nil} &
4\colon\m{g}(x \app \m{nil}) &\to \m{g}(x) \app \m{nil} &
6\colon\m{g}(x \app (y \app z)) &\to \m{g}((x \app y) \app z)
\end{xalignat*}
In~\cite{MSW08} a TMI compatible with~\TRS{R} of
dimension four is given showing that the derivational complexity is
bounded by a polynomial of degree four. Using Theorem~\ref{THM:modeq}
with TMIs of dimension three yields a cubic upper bound, i.e.,
the TMI~$\ALG{M}_1$
\begin{xalignat*}{2}
\m{\circ}_{\ALG{M}_1}(\vec x,\vec y) &= \BM
1\NE 0\NE 0\NR
0\NE 0\NE 1\NR
0\NE 0\NE 1
\EM \vec x +
\vec y +
\BM
0\NR
0\NR
1\NR
\EM
&
\m{f}_{\ALG{M}_1}(\vec x) &= \BM
1\NE 1\NE 0\NR
0\NE 0\NE 1\NR
0\NE 0\NE 1
\EM \vec x +
\BM
0\NR
1\NR
0\NR
\EM
\\
\m{g}_{\ALG{M}_1}(\vec x) &= \BM
1\NE 0\NE 0\NR
0\NE 0\NE 1\NR
0\NE 0\NE 1
\EM \vec x
&
\m{nil}_{\ALG{M}_1} &= \BM
1\NR
1\NR
1\NR
\EM
\end{xalignat*}
yields a cubic upper bound on $\REL{\{1,2,3\}}{\{4,5,6\}}$.
So does the TMI~$\ALG{M}_2$ with
\begin{xalignat*}{4}
\m{f}_{\ALG{M}_2}(\vec x) &= \m{g}_{\ALG{M}_1}(\vec x)&
\m{g}_{\ALG{M}_2}(\vec x) &= \m{f}_{\ALG{M}_1}(\vec x)&
\m{\circ}_{\ALG{M}_2}(\vec x,\vec y) &= \m{\circ}_{\ALG{M}_1}(\vec y,\vec x)&
\m{nil}_{\ALG{M}_2} &= \m{nil}_{\ALG{M}_1}
\end{xalignat*}
for $\REL{\{4,5,6\}}{\{1,2,3\}}$. Our approach enables showing a lower
complexity than~\cite{MSW08} but the derivational complexity
of~\TRS{R} is quadratic (see~\cite{MSW08}).
The quadratic lower bound is justified as~\TRS{R} admits derivations
\[
\m{f}^n(x)\sigma^m \to^n
\m{nil} \circ \m{f}^n(x)\sigma^{m-1} \to^n
\cdots \to^n
x\sigma^m\tau^n
\]
of length $nm$ where
$\sigma = \{x \mapsto \m{nil} \circ x\}$ and
$\tau = \{x \mapsto \m{f}(x)\}$.
We stress that the recent approach in~\cite{W10} allows to establish a
quadratic upper bound. For a comment on the integration of this method
into our setting we refer to Section~\ref{CON:main}.
\end{exa}

\section{Implementation}
\label{IMP:main}

In Section~\ref{IMP:est} we first show how the various theorems from the
previous sections can be implemented to obtain \emph{some} complexity proof.
Afterwards Section~\ref{IMP:imp} is concerned
with lowering the bounds starting from an existing complexity proof.

\subsection{Establishing Bounds}
\label{IMP:est}

To estimate the complexity of a TRS $\TRS{R}$ with respect to a
language~$L$, we first transform $\TRS{R}$ into the relative TRS
$\REL{\TRS{R}}{\varnothing}$. Obviously
$
\cp{n}{\to_\TRS{R}}{L} =
\cp{n}{\to_\REL{\TRS{R}}{\varnothing}}{L}
$.
If the input already is a relative TRS this step is omitted.
Afterwards for a relative TRS~\REL{\TRS{R}}{\TRS{S}} we try to
establish a bound on the complexity of
$\REL{\TRS{R}_2}{(\TRS{R}_1\cup\TRS{S})}$ with respect to~$L$
for some $\TRS{R}_1$, $\TRS{R}_2$ with
$\TRS{R}_1 = \TRS{R} \setminus \TRS{R}_2$
and continue with the relative TRS
$\REL{\TRS{R}_1}{(\TRS{R}_2\cup\TRS{S})}$.
This step is executed repeatedly until the remaining problem equals
$\REL{\varnothing}{(\TRS{R}\cup\TRS{S})}$. Then the complexity of
$\REL{\TRS{R}}{\TRS{S}}$ with respect to~$L$ is obtained by summing up
all intermediate bounds. In order to establish a maximal number of
complexity proofs we run all techniques from Sections~\ref{MAT:main}
and~\ref{BOUNDS:main} in parallel and the first technique that can
shift some rules is used to achieve progress.

Note that the procedure sketched above contains an implicit application of
Theorem~\ref{THM:modeq}, i.e., some method immediately proves a bound
for $\REL{\TRS{R}_2}{(\TRS{R}_1\cup\TRS{S})}$ and leaves
$\REL{\TRS{R}_1}{(\TRS{R}_2\cup\TRS{S})}$ as open proof obligation.
In contrast to an explicit application of Theorem~\ref{THM:modeq},
here the method that establishes the bound on
$\REL{\TRS{R}_2}{(\TRS{R}_1\cup\TRS{S})}$ can select the decomposition
of~\TRS{R} into $\TRS{R}_1$ and $\TRS{R}_2$ which is beneficial for
performance. As an immediate consequence, proof trees degenerate to
lists (cf.\ Example~\ref{EX:linear}). In the following we describe
the presented approach more formal and refer to it as the
\emph{complexity framework}.

\begin{defi}
A \emph{complexity problem} (CP problem for short) is a pair
$(\REL{\TRS{R}}{\TRS{S}},L)$ consisting of a relative TRS
$\REL{\TRS{R}}{\TRS{S}}$ and a language~$L$.
\end{defi}

To operate on CP problems so called \emph{complexity processors} are
used. Similar as in the dependency pair framework we distinguish between
sound and complete processors. Here sound complexity processors
are used to prove an upper bound on the complexity of a given CP problem
whereas complete complexity processors are applied to derive
lower bounds on the complexity.

\begin{defi}
\label{DEF:CP processor}
A \emph{complexity processor} (CP processor for short) is a function
that takes a CP problem $(\REL{\TRS{R}}{\TRS{S}},L)$ as input and returns
a set of pairs
$\bigcup_{1\leqslant i\leqslant m}\{ ((\REL{\TRS{R}_i}{\TRS{S}_i},L_i),f_i)\}$
as output.%
\footnote{\
For reasons of readability we write pairs
$((\REL{\TRS{R}_i}{\TRS{S}_i},L_i),f_i)$
as triples
$(\REL{\TRS{R}_i}{\TRS{S}_i},L_i,f_i)$.
}
Here $(\REL{\TRS{R}_i}{\TRS{S}_i},L_i)$ is a complexity problem
and $f_i \colon \Nat \to \Nat$ for each $1\leqslant i\leqslant m$.
A complexity processor is \emph{sound} if
\[
\cp{n}{\to_{\REL{\TRS{R}}{\TRS{S}}}}{L} =
\OO(
f_1(n) + \cdots + f_m(n)
+ \cp{n}{\to_{\REL{\TRS{R}_1}{\TRS{S}_1}}}{L_1} + \cdots
+ \cp{n}{\to_{\REL{\TRS{R}_m}{\TRS{S}_m}}}{L_m}
)
\]
and it is called \emph{complete} if
\[
\cp{n}{\to_{\REL{\TRS{R}}{\TRS{S}}}}{L} =
\Omega(
f_1(n) + \cdots + f_m(n)
+ \cp{n}{\to_{\REL{\TRS{R}_1}{\TRS{S}_1}}}{L_1} + \cdots
+ \cp{n}{\to_{\REL{\TRS{R}_m}{\TRS{S}_m}}}{L_m}
)
\]
holds.
\end{defi}

In the sequel $\zero$ denotes the constant zero function,
i.e., $\zero\colon \Nat \to \Nat$ with $\zero(n) = 0$.
Next we list some CP processors that can be derived from the previous
sections. The first one is based on complexity pairs and can, e.g., be
implemented by Theorems~\ref{THM:tmi} and~\ref{THM:ami}.

\begin{thm}
\label{THM:cp processor}
The CP processor
\[
(\REL{\TRS{R}}{\TRS{S}},L) \mapsto
\begin{cases}
\{(\REL{\TRS{R}_1}{(\TRS{R}_2 \cup \TRS{S})},L,f)\}
& \text{if $\REL{\TRS{R}_2}{(\TRS{R}_1\cup\TRS{S})}$ is compatible}\\[-0.5ex]
& \text{with a complexity pair $(\relgt,\relge)$}\\
\{(\REL{\TRS{R}}{\TRS{S}},L,\zero)\} & \text{otherwise}
\end{cases}
\]
where $\TRS{R} = \TRS{R}_1 \cup \TRS{R}_2$,
and $f(n) = \cp{n}{\relgt}{L}$ is sound.
\end{thm}
\begin{proof}
Follows from Corollary~\ref{COR:bound} and Theorem~\ref{THM:modeq}.
\end{proof}

The above processor is implemented by demanding that all rules in
$\TRS{R}\cup\TRS{S}$ are weakly oriented while at least one rule
in~\TRS{R} is oriented strictly. Hence the decomposition of~\TRS{R}
into~$\TRS{R}_1$ and~$\TRS{R}_2$ is performed automatically.
The next CP processor requires a mild condition on $\TRS{R}_1$ only.
Again the decomposition of~$\TRS{R}$ into~$\TRS{R}_1$ and~$\TRS{R}_2$ is
performed automatically since in the implementation we just demand that
the $\TRS{S}$-rules are weakly oriented while 
the \TRS{R}-rules may increase by a constant factor and
at least one of the rules in~\TRS{R} is oriented strictly.

\begin{thm}
\label{THM:sli processor}
The CP processor
\[
(\REL{\TRS{R}}{\TRS{S}},L) \mapsto
\begin{cases}
\{(\REL{\TRS{R}_1}{(\TRS{R}_2 \cup \TRS{S})},L,f)\}
& \text{if~\ALG{M} is a matrix interpretation with constant}\\[-.5ex]
& \text{growth, $\TRS{R}_1 \subseteq {\relge^\ncp_\ALG{M}}$ and
        $\REL{\TRS{R}_2}{\TRS{S}}$ is compatible with~\ALG{M}}\\
\{(\REL{\TRS{R}}{\TRS{S}},L,\zero)\} & \text{otherwise}
\end{cases}
\]
where $\TRS{R} = \TRS{R}_1 \cup \TRS{R}_2$, and $f(n) = n$ is sound.
\end{thm}
\begin{proof}
Follows from Theorem~\ref{THM:modeq} and Theorem~\ref{THM:cgp_tmi}.
\end{proof}

The next CP processor is based on match-bounds.

\begin{thm}
\label{THM:match(-raise)-RT-bounds processor}
The CP processor
\[
(\REL{\TRS{R}}{\TRS{S}}) \mapsto
\begin{cases}
\{(\REL{\TRS{R}_1}{(\TRS{R}_2 \cup \TRS{S})},L,f)\}
& \text{if $\REL{\TRS{R}_2}{(\TRS{R}_1 \cup \TRS{S})}$ is}\\[-0.5ex]
& \text{linear and match-RT-bounded for~L or}\\[-0.5ex]
& \text{non-duplicating and match-raise-RT-bounded for~$L$}\\
\{(\REL{\TRS{R}}{\TRS{S}},L,\zero)\} & \text{otherwise}
\end{cases}
\]
where $\TRS{R} = \TRS{R}_1 \cup \TRS{R}_2$,
$\TRS{R}_2$ is non-collapsing, and $f(n) = n$ is sound.
\end{thm}
\begin{proof}
Follows from Theorems~\ref{THM:modeq},%
~\ref{THM:match-RT-bounds => linear complexity}, and%
~\ref{THM:match-raise-RT-bounds => linear complexity}.
\end{proof}

The above processor is implemented by considering for any
rule $l \to r \in \TRS{R}$ the decompositions
$\TRS{R}_2 = \{l \to r\}$ and $\TRS{R}_1 = \TRS{R}\setminus \TRS{R}_2$
in parallel.
The next CP processor we present is not implemented for finding
a bound (cf.\ the discussion at the beginning of the section) but
very suitable to tighten existing bounds (see Section~\ref{IMP:imp}).

\begin{thm}
\label{THM:modeq processor}
The CP processor
\[
(\REL{\TRS{R}}{\TRS{S}},L) \mapsto
\{
 (\REL{\TRS{R}_1}{(\TRS{R}_2 \cup \TRS{S})},L,\zero),
 (\REL{\TRS{R}_2}{(\TRS{R}_1 \cup \TRS{S})},L,\zero)
\}
\]
where $\TRS{R} = \TRS{R}_1 \cup \TRS{R}_2$ is sound and complete.
\end{thm}
\begin{proof}
By Theorem~\ref{THM:modeq}.
\end{proof}

Finally, the main theorem states that the CP framework may be applied to
complexity analysis. We say that $P$ is a complexity proof for a relative
TRS $\REL{\TRS{R}}{\TRS{S}}$ and a language~$L$ if all leaves in
$P$ are of the shape $\REL{\varnothing}{(\TRS{R}\cup\TRS{S})}$.

\begin{thm}
\label{THM:imp}
Let $\REL{\TRS{R}}{\TRS{S}}$ be a relative TRS and~$L$ be a language.
Let $P$ be a complexity proof for $\REL{\TRS{R}}{\TRS{S}}$ and~$L$
and $f_1, \ldots, f_m$ be the complexities occurring in this proof.
If all CP processors in $P$ are sound then
$\cp{n}{\to_\REL{\TRS{R}}{\TRS{S}}}{L} = \OO(f_1(n) + \cdots + f_m(n))$.
Similarly, if all CP processors in $P$ are complete then
$\cp{n}{\to_\REL{\TRS{R}}{\TRS{S}}}{L} = \Omega(f_1(n) + \cdots + f_m(n))$.
\end{thm}
\begin{proof}
By Definition~\ref{DEF:CP processor} as well as basic
properties of $\OO$-notation.
\end{proof}

We conclude the section with an (abstract) example which illustrates the
behavior of the complexity framework.

\begin{exa}
\label{EX:linear}
Consider the TRS $\TRS{R}$ of Example~\ref{EX:reverse} on
page~\pageref{EX:reverse} with the complexity proof depicted in
Figure~\ref{FIG:linear}. After transforming $\TRS{R}$ into the
relative TRS $\REL{\TRS{R}}{\varnothing}$ the CP processor of
Theorem~\ref{THM:cp processor} is applied twice. First the
(derivational) complexity of the relative TRS $\REL{\{1,3,5\}}{\{2,4\}}$
is estimated by a polynomial of degree five. As a consequence, the
rules 1, 3, and 5 are moved into the relative component yielding a
CP problem consisting of the relative TRS $\REL{\{2,4\}}{\{1,3,5\}}$.
After that the (derivational) complexity of
$\REL{\{2,4\}}{\{1,3,5\}}$ is estimated by a quadratic bound.
Since the remaining CP problem is of the shape
$\REL{\varnothing}{\TRS{R}}$ according to Theorem~\ref{THM:imp}
the (derivational) complexity of $\TRS{R}$ is at most quintic.
\begin{figure}
\begin{tikzpicture}[node distance=12mm and 14mm,on grid]
\node (1)                  {$\TRS{R}$};
\node (2)   [below=of 1]   {$\REL{\TRS{R}}{\varnothing}$};
\node (22)  [below=of 2]   {$\REL{\{2,4\}}{\{1,3,5\}}$};
\node (222) [below=of 22]  {$\REL{\varnothing}{\TRS{R}}$};
\draw[->] (1)  to node[right] {{\scriptsize$\OO(1)$}}   (2);
\draw[->] (2)  to node[right] {{\scriptsize$\OO(n^5)$}} (22);
\draw[->] (22) to node[right] {{\scriptsize$\OO(n^2)$}} (222);
\end{tikzpicture}
\caption{Linear complexity proof}
\label{FIG:linear}
\end{figure}
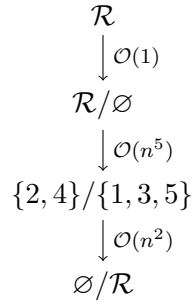
\end{exa}

\subsection{Tightening Bounds}
\label{IMP:imp}

In contrast to termination, which is a plain YES/NO question, complexity
corresponds to an optimization problem. Hence automated tools should try to
establish as tight bounds as possible. In the direct setting all complexity
methods can be executed in parallel and after a fixed amount of time the
tightest bound is reported. The next example shows such a case.

\begin{exa}
\label{EX:direct}
Consider the TRS~\BITS (\tpdb{nontermin/AG01/\#4.28})
consisting of the following five rules:
\begin{alignat*}{4}
1\colon&& \m{half}(\m{0}) &\to \m{0} & \qquad \qquad
4\colon&& \m{bits}(\m{0}) &\to \m{0}\\
2\colon&& \m{half}(\m{s}(\m{0})) &\to \m{0} &
5\colon&& \m{bits}(\m{s}(x)) &\to \m{s}(\m{bits}(\m{half}(\m{s}(x)))) \\
3\colon&& \m{half}(\m{s}(\m{s}(x))) &\to \m{s}(\m{half}(x))
\end{alignat*}
For this TRS the complexity analyzer \CAT (cf.\ Section~\ref{EXP:main})
finds a proof by root-labeling followed by a TMI of dimension two,
establishing a quadratic upper bound 
within five seconds.
However, after 90 seconds the tool finds the following AMI~$\ALG{A}$
that shows a linear upper bound:
\begin{xalignat*}{4}
\m{bits}_\ALG{A}(\vec x) &= \BM
0\NE 1\NE 2\NR
0\NE 4\NE 5\NR
0\NE 6\NE 7
\EM \vec x
&
\m{half}_\ALG{A}(\vec x) &= \BM
1\NE \minfty\NE \minfty\NR
1\NE \minfty\NE \minfty\NR
1\NE \minfty\NE \minfty
\EM \vec x
&
\m{s}_\ALG{A}(\vec x) &= \BM
1\NE 1\NE \minfty\NR
7\NE 0\NE 2\NR
1\NE 6 \NE 0
\EM \vec x
&
\m{0}_\ALG{A} &= \BM
 0\NR
 0\NR
 \minfty
\EM
\end{xalignat*}
So, whenever the tool is allowed more than 90 seconds the linear bound
can be reported and if the user sets the global timeout to less, then
still the quadratic bound can be output.
\end{exa}

In the modular setting this simple idea does not work because
two problems emerge. The first problem is that
the tool does not know how much time it may spend
in a single proof step. If it spends too much then it may not finish the
proof within the global time limit and if it spends too little then it
can miss a low bound.
The second problem is that in the modular setting
separate criteria may make statements about the complexity of different
rules. The question is then to identify the \emph{better} bound.
The next example demonstrates this scenario.

\begin{exa}
Consider the TRSs
\begin{xalignat*}{2}
1\colon \m{a}(\m{b}(x)) &\to \m{b}(\m{a}(x)) &
2\colon \m{g}(x,x) &\to \m{g}(\m{c},\m{d})
\end{xalignat*}
The TMI~$\ALG{M}$ with $\m{g}_\ALG{M}(\vec x,\vec y) = \vec x + \vec y$
and
\begin{xalignat*}{4}
\m{a}_\ALG{M}(\vec x) &= \BM
1\NE 1\NR
0\NE 1\NR
\EM \vec x
&
\m{b}_\ALG{M}(\vec x) &= \BM
1\NE 0\NR
0\NE 1\NR
\EM \vec x + \BM
0\NR
1\NR
\EM
&
\m{c}_\ALG{M}(\vec x) &= \BM
0\NR
0\NR
\EM
&
\m{d}_\ALG{M}(\vec x) &= \BM
0\NR
0\NR
\EM
\end{xalignat*}
establishes a quadratic upper bound on the complexity of
$\REL{\{1\}}{\{2\}}$
whereas match-bounds yield a linear upper bound for
$\REL{\{2\}}{\{1\}}$.
The question now is with which remaining proof obligation
($\REL{\{2\}}{\{1\}}$ or $\REL{\{1\}}{\{2\}}$)
the tool should continue. Note that both bounds are tight.
\end{exa}

The following idea overcomes both problems:
First we establish \emph{some} complexity proof according to the procedure
described at the beginning of Section~\ref{IMP:est} to obtain a bound for
as many systems as possible. Afterwards we \emph{optimize} this bound.
The next example shows how the latter works.

\begin{exa}
\label{EX:opt}
Consider the TRS $\TRS{R}$ of Example~\ref{EX:reverse} with the complexity
proof depicted in Figure~\ref{FIG:original}. In this exemplary
case one part in this proof, highlighted by a solid box, is overestimated
by a cubic upper bound. Hence the complexity of the whole system is at most
cubic. We remark that this proof step estimates the complexity of
$\REL{\{3,5\}}{\{1,2,4\}}$. Now assume that the cubic bound
is not optimal, i.e., there exists a proof (that may be longer and
harder to find) that induces a quadratic upper bound on the complexity of
$\REL{\{3,5\}}{\{1,2,4\}}$. Then the proof is optimized as illustrated
in Figure~\ref{FIG:optimized}, i.e., $\REL{\{1,3,5\}}{\{2,4\}}$ is
split into the problems
$\REL{\{1\}}{\{2,3,4,5\}}$ and $\REL{\{3,5\}}{\{1,2,4\}}$
by an application of Theorem~\ref{THM:modeq}.
After that, the proof part of $\REL{\{1\}}{\{2,3,4,5\}}$ is reused in the
optimized proof (cf.\ the dashed boxes in Figure~\ref{FIG:original} and
Figure~\ref{FIG:optimized}) whereas the original proof of
$\REL{\{3,5\}}{\{1,2,4\}}$ is replaced by the new one, as
indicated by the solid box in Figure~\ref{FIG:optimized}.
Now, the proof in Figure~\ref{FIG:optimized} establishes a quadratic
upper bound on the complexity of~\TRS{R}. For completeness we state 
that the proof in Figure~\ref{FIG:original} can be obtained from the
linear proof tree shown in Figure~\ref{FIG:linear} by optimization. To show
the procedure on a non-linear proof tree this presentation was chosen.
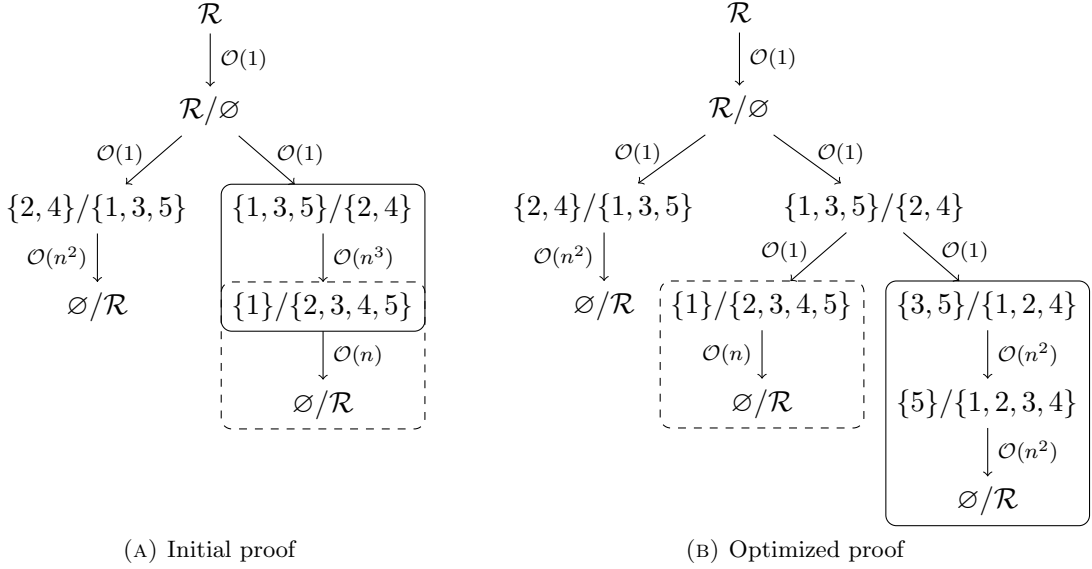
\begin{figure}
\subfloat[Initial proof]{
\label{FIG:original}
\begin{tikzpicture}[node distance=13mm and 15mm,on grid]
\node (1)                         {$\TRS{R}$};
\node (2)   [below=of 1]          {$\REL{\TRS{R}}{\varnothing}$};
\node (21)  [below left=of 2]     {$\REL{\{2,4\}}{\{1,3,5\}}$};
\node (211) [below=of 21]         {$\REL{\varnothing}{\TRS{R}}$};
\node (22)  [below right=of 2]    {$\REL{\{1,3,5\}}{\{2,4\}}$};
\node (221) [below=of 22]         {$\REL{\{1\}}{\{2,3,4,5\}}$};
\node (222) [below=of 221]        {$\REL{\varnothing}{\TRS{R}}$};
\node (223) [below=of 222]        {\phantom{$\REL{\varnothing}{\TRS{R}}$}};
\draw[->] (1)   to node[right]      {{\scriptsize$\OO(1)$}}   (2);
\draw[->] (2)   to node[base left]  {{\scriptsize$\OO(1)$}}   (21);
\draw[->] (2)   to node[base right] {{\scriptsize$\OO(1)$}}   (22);
\draw[->] (21)  to node[left]       {{\scriptsize$\OO(n^2)$}} (211);
\draw[->] (22)  to node[right]      {{\scriptsize$\OO(n^3)$}} (221);
\draw[->] (221) to node[right]      {{\scriptsize$\OO(n)$}}   (222);
\node[draw,inner sep=0mm,rectangle,rounded corners,fit=(22) (221)]         {};
\node[draw,inner sep=0mm,rectangle,rounded corners,dashed,fit=(221) (222)] {};
\end{tikzpicture}
}
\qquad
\subfloat[Optimized proof]{
\label{FIG:optimized}
\begin{tikzpicture}[node distance=13mm and 15mm,on grid]
\node (1)                                                   {$\TRS{R}$};
\node (2)    [below=of 1]                                   {$\REL{\TRS{R}}{\varnothing}$};
\node (21)   [node distance=13mm and 18mm,below left=of 2]  {$\REL{\{2,4\}}{\{1,3,5\}}$};
\node (211)  [below=of 21]                                  {$\REL{\varnothing}{\TRS{R}}$};
\node (22)   [node distance=13mm and 18mm,below right=of 2] {$\REL{\{1,3,5\}}{\{2,4\}}$};
\node (221)  [below left=of 22]                             {$\REL{\{1\}}{\{2,3,4,5\}}$};
\node (2211) [below=of 221]                                 {$\REL{\varnothing}{\TRS{R}}$};
\node (222)  [below right=of 22]                            {$\REL{\{3,5\}}{\{1,2,4\}}$};
\node (2221) [below=of 222]                                 {$\REL{\{5\}}{\{1,2,3,4\}}$};
\node (2222) [below=of 2221]                                {$\REL{\varnothing}{\TRS{R}}$};
\draw[->] (1)    to node[right]      {{\scriptsize$\OO(1)$}}   (2);
\draw[->] (2)    to node[base left]  {{\scriptsize$\OO(1)$}}   (21);
\draw[->] (2)    to node[base right] {{\scriptsize$\OO(1)$}}   (22);
\draw[->] (21)   to node[left]       {{\scriptsize$\OO(n^2)$}} (211);
\draw[->] (22)   to node[base left]  {{\scriptsize$\OO(1)$}}   (221);
\draw[->] (22)   to node[base right] {{\scriptsize$\OO(1)$}}   (222);
\draw[->] (221)  to node[left]       {{\scriptsize$\OO(n)$}}   (2211);
\draw[->] (222)  to node[right]      {{\scriptsize$\OO(n^2)$}} (2221);
\draw[->] (2221) to node[right]      {{\scriptsize$\OO(n^2)$}} (2222);
\node[draw,inner sep=0mm,rectangle,rounded corners,fit=(222) (2221) (2222)] {};
\node[draw,inner sep=0mm,rectangle,rounded corners,dashed,fit=(221) (2211)] {};
\end{tikzpicture}
}
\caption{Tightening bounds}
\end{figure}
\end{exa}

As the previous example demonstrates the basic idea is to replace
single proof steps by new proofs that induce tighter bounds. This
procedure is repeated until either the global time limit is reached or
none of the bounds can be tightened further. Note that the transformation is
sound by Theorem~\ref{THM:imp}.

The final example in this section shows that it may be easier to find
proofs in the modular setting.
\begin{exa}
Recall the TRS from Example~\ref{EX:direct}. Corollary~\ref{COR:sli} with
an SLI that just counts function symbols allows to transform the initial
problem into $\REL{\{5\}}{\{1,2,3,4\}}$. The AMI of dimension three with
\begin{xalignat*}{4}
\m{bits}_\ALG{A}(\vec x) &= \BM
0\NE 0\NE 0\NR
0\NE 3\NE 2\NR
0\NE 2\NE 2\NR
\EM \vec x
&
\m{half}_\ALG{A}(\vec x) &= \BM
0\NE \minfty\NE \minfty\NR
0\NE \minfty\NE \minfty\NR
0\NE \minfty\NE \minfty\NR
\EM \vec x
&
\m{s}_\ALG{A}(\vec x) &= \BM
0\NE \minfty\NE 0\NR
\minfty\NE \minfty\NE 3\NR
4\NE 0 \NE 0\NR
\EM \vec x
&
\m{0}_\ALG{A} &= \BM
 0\NR
 0\NR
 0\NR
\EM
\end{xalignat*}
allows to show linear derivational complexity of~\BITS. Note
that \CAT finds this interpretation within three seconds whereas it took
the tool 90 seconds to find a suitable interpretation for the direct setting.
\end{exa}

\section{Experimental Results}
\label{EXP:main}

The techniques described in the preceding sections are implemented in
the complexity analyzer \CAT (freely available from
\url{http://cl-informatik.uibk.ac.at/software/cat}) which is built on
top of \TTTT~\cite{KSZM09}, a powerful termination tool for TRSs.

Below we report on the experiments%
\footnote{\
\label{FOO:web}
Full details available from
\url{http://cl-informatik.uibk.ac.at/software/cat/10lmcs}.
}
we performed.
We considered the \TRSALLNUM TRSs in version 7.0.2
of the TPDB without strategy or theory annotation. The \TRSDCNUM
non-duplicating systems of this collection have been used for experiments
with derivational complexity (note that a duplicating system gives rise to
at least exponentially long derivations). For runtime complexity we
considered the \TRSRCNUM systems that are not trivial, e.g., where the
set of constructor-based terms is not finite and terminating.
In this collection there are \TRSRCNDNUM non-duplicating systems.
All tests have been performed on a server equipped with
eight dual-core AMD Opteron\textsuperscript{\textregistered}\xspace
processors 885 running at a clock rate of 2.6~GHz and 64~GB of main
memory. We remark that similar results have been obtained on a
dual-core laptop. If a tool did not report an answer within 60
seconds, its execution was aborted.

As complexity preserving transformations we employ
uncurrying~\cite{ZHM10} for applicative systems whenever it
applies and root-labeling~\cite{SM08} in parallel to the base
methods.
As base methods we use the match-bounds technique as well as
TMIs~\cite{MSW08,NZM10} and AMIs~\cite{KW09} of dimensions one to five.
The latter two are implemented by bit-blasting arithmetic operations to
SAT~\cite{EWZ08}.
All base methods are run in parallel and started upon program execution.

Our results are summarized in Tables~\ref{TAB:trs} and~\ref{TAB:trs_rc}.
Here, \direct refers to the conventional setting where all rules must be
oriented at once whereas \modular first transforms a TRS~\TRS{R} into a
relative TRS~\REL{\TRS{R}}{\varnothing} before the CP processors from
Section~\ref{IMP:est} (except Theorem~\ref{THM:modeq processor})
are employed. In the tables the postfix $\star$ indicates that after
establishing a bound it is tried to be tightened as explained in
Section~\ref{IMP:imp}.
The columns \linear, \quadratic, \dots, \poly give the number of
linear, quadratic, \dots, polynomial upper bounds that could be
established.
We also list the average time (in seconds) needed for
finding a bound in the last column.
For reference we also give the data for the winners of the corresponding 
categories in the 2010 edition of the termination competition.
For derivational complexity this is~\CAT and for runtime complexity
this is~\TCT~\cite{AMS08}.

Table~\ref{TAB:trs} shows the results for derivational complexity. Here
the modular approach allows to prove significantly more polynomial bounds
(column \poly) and furthermore these bounds are also smaller than for the
direct approach (especially if tightening of bounds is used). The modular
setting is slower since there typically more proofs are required to succeed.
The rows postfixed $\star$ prove that refining bounds is beneficial,
especially if all criteria are run in parallel, which is essential to
maximize the total number of upper bounds. The 2010 version of \CAT did
not use tightening of bounds. To maximize the number of low bounds the
tool executes criteria that yield larger complexity bounds slightly delayed.
This explains why for \CAT tightening bounds increases the global
performance less compared to \direct and \modular. On the contrary, \CAT
misses some proofs compared to \modular since (costly) criteria are not
executed for up to 60 seconds.

\begin{table}[t]
\caption{Derivational complexity of \TRSDCNUM~TRSs}
\label{TAB:trs}
\begin{center}
\begin{tabular}{@{}l@{\qquad}
 cccc@{\qquad}
 r
 @{}
}
\hline
& \poly & \linear & \quadratic & \cubic & \avg \\[0.1ex]
\hline
\direct        & 315 & 202 & 234 & 259 & 1.7 \\
$\direct\star$ & 315 & 215 & 303 & 312 & 7.9 \\
\modular       & 334 & 208 & 228 & 261 & 2.6 \\
$\modular\star$& 334 & 221 & 321 & 329 & 10.6 \\
\hline
$\CAT$     & 328 & 216 & 310 & 319 & 4.2 \\
$\CAT\star$& 328 & 219 & 317 & 324 & 11.5 \\
\end{tabular}
\end{center}
\end{table}

Table~\ref{TAB:trs_rc} shows the results for runtime complexity on the
\TRSRCNDNUM TRSs that are non-duplicating and non-trivial. Here,
the starting language for the match-bounds technique has been restricted
to constructor-based terms, i.e., no defined symbols are allowed below the
root. This makes match-bounds a very powerful technique for runtime
complexity, explaining the high number of linear bounds.%
\footnote{\
In the 2010 competition \CAT proved upper bounds on the derivational
complexity also in the category for runtime complexity. This explains why
our methods here outperform \TCT while in the competition \CAT came
second in this division.}
We remark that in contrast to the criteria we employ \TCT can
also estimate polynomial bounds for duplicating systems based on weak
dependency pairs~\cite{HM08,HM08b}. The row
$\TCT(\TRSRCNUM)$ corresponds to \TCT run on all \TRSRCNUM TRSs
in the benchmark for runtime complexity.
Hence this row includes duplicating TRSs. For 9 of these
\TCT can prove a polynomial upper bound.

\begin{table}[t]
\caption{Runtime complexity of \TRSRCNDNUM~TRSs}
\label{TAB:trs_rc}
\begin{center}
\begin{tabular}{@{}l@{\qquad}
 cccc@{\qquad}
 r
 @{}
}
\hline
& \poly & \linear & \quadratic & \cubic & \avg \\[0.1ex]
\hline
\direct        & 372 &  354 & 358 & 365 & 0.6 \\
$\direct\star$ & 372 &  355 & 370 & 371 & 1.1 \\
\modular       & 376 &  358 & 359 & 364 & 0.5 \\
$\modular\star$& 376 &  362 & 374 & 375 & 1.2 \\
\hline
$\TCT$            & 354 & 351 & 353 & 354 & 4.8 \\
$\TCT(\TRSRCNUM)$ & 363 & 360 & 362 & 363 & 4.8 \\
\end{tabular}
\end{center}
\end{table}

For further comparison with other tools we refer the reader to the
international termination competition (referenced in Footnote~\ref{FOO:comp}
on page~\pageref{FOO:comp}).
Since 2008, when the complexity categories have been installed in the
termination competition, \CAT won the division for derivational complexity
every year.

\section{Conclusion}
\label{CON:main}

In this article we have introduced a modular approach for estimating the
complexity of TRSs by considering relative rewriting. We
showed how existing criteria (for full rewriting) can be lifted into the
relative setting. The modular approach is easy to
implement and has been proved strictly more powerful than traditional methods
in theory and practice. Since the modular method allows to combine different
criteria, typically smaller complexity bounds are achieved
than with the direct setting. Furthermore the modular treatment allows
to establish bounds for systems where each of the involved basic methods
alone fails.
Although originally developed for derivational complexity our
results directly apply to more restrictive notions of complexity,
e.g., runtime complexity (see also below).
Finally we remark that our setting allows a more fine-grained complexity
analysis, i.e., while traditionally quadratic derivational complexity
ensures that any rule is applied at most quadratically often, our approach
can make different statements about single rules. Hence even if a proof
attempt does not succeed completely, it may highlight the problematic rules.

We remark that complexity proofs using TMIs (for relative rewriting) can be
certified with \CETA~\cite{TS09b}.

\medskip

As related work we mention~\cite{HW06t} which also considers relative
rewriting for complexity analysis. However, there the complexity of
$\TRS{R}_1\cup\TRS{R}_2$ is investigated by considering
\REL{\TRS{R}_1}{\TRS{R}_2} and $\TRS{R}_2$. Hence~\cite{HW06t} also gives
rise to a modular reasoning but the obtained complexities are typically
beyond polynomials.
For runtime complexity analysis Hirokawa and Moser~\cite{HM08,HM08b}
consider weak dependency pair steps relative to the usable rules, i.e.,
\REL{\WDP(\TRS{R})}{\UR(\TRS{R})}.
However, since in the current formulation of weak dependency pairs
some complexity might be hidden in the usable rules they do not
really obtain a relative problem. As a consequence they can only apply
restricted criteria for the usable rules. Note that our approach can
directly be used to show bounds on \REL{\WDP(\TRS{R})}{\UR(\TRS{R})} by
considering~$\WDP(\TRS{R})\cup\UR(\TRS{R})$. Due to
Corollary~\ref{COR:sli} this problem can be transformed into an
(unrestricted) relative problem \REL{\WDP(\TRS{R})}{\UR(\TRS{R})} whenever
the constraints in~\cite{HM08} are satisfied.
Moreover, if somehow the \emph{problematic} usable rules could be determined
and shifted into the $\WDP(\TRS{R})$ component, then this
\emph{improved} version
of weak dependency pairs corresponds to a relative problem without
additional restrictions, admitting further benefit from our
contributions.

Recently two approaches were proposed which admit polynomially bounded
matrix interpretations going beyond TMIs. While~\cite{W10} considers
weighted automata, in~\cite{NZM10} (joint) spectral radius theory is employed.
For ease of presentation these criteria have not been considered in this
work but since both are based on matrix interpretations, they perfectly
suit our modular setting. 

\medskip 
For future work we plan to investigate criteria that allow to analyze
the complexity of a TRS~\TRS{R} by the complexities of $\TRS{R}_1$ and
$\TRS{R}_2$ where $\TRS{R} = \TRS{R}_1 \cup \TRS{R}_2$. We anticipate
that results from modularity~\cite{O94} are helpful for this aim.

\subsubsection*{Acknowledgments}

We thank Johannes Waldmann for directing our attention to
Example~\ref{EX:counterexample} and Martin Avanzini for providing a
binary of the 2010 competition version of~\TCT.


\bibliographystyle{lncs}
\selectlanguage{english}
\bibliography{references}

\end{document}